\documentclass[11pt]{article}

\usepackage[english]{babel}
\usepackage{amsmath}
\usepackage{amssymb}
 \usepackage{amstext,amsthm,bbm}
\usepackage{graphicx}
\usepackage{color}

\usepackage{enumerate}
\usepackage{fullpage}

\newtheorem{theorem}{Theorem}[section]
\newtheorem{corollary}[theorem]{Corollary}
\newtheorem{proposition}[theorem]{Proposition}

\theoremstyle{definition}

\newtheorem{remark}[theorem]{Remark}

\numberwithin{equation}{section}

\newcommand{\R}{\mathbb R}

\newcommand{\C}{\mathbb C}

\renewcommand{\det}{{\,\rm det}\:}

\newcommand{\Tr}{{\,\rm Tr}\:}

\newcommand{\erfc}{{\,\rm erfc}\:}

\renewcommand{\exp}[1]{\mathrm{exp} \left\{#1\right\}}

\begin{document}

\title{\bf On statistics of bi-orthogonal eigenvectors in real and complex Ginibre ensembles: combining partial Schur decomposition with supersymmetry. }

\author{\Large Yan V Fyodorov	\footnote{yan.fyodorov@kcl.ac.uk}\\[0.5ex] {\small Department of Mathematics, King's College London}\\
{\small Strand, London, WC2R 2LS, UK}}

\date{05 March 2018}

\maketitle

\begin{abstract}
	We suggest a method of studying the joint probability density (JPD) of an eigenvalue and the associated 'non-orthogonality overlap factor' (also known as the 'eigenvalue condition number') of the  left and right eigenvectors for non-selfadjoint Gaussian random matrices of size $N\times N$. First we derive the general finite $N$ expression for the JPD of a real eigenvalue $\lambda$ and the associated non-orthogonality factor in the real Ginibre ensemble, and then analyze its 'bulk' and 'edge' scaling limits. The ensuing distribution is maximally heavy-tailed, so that all integer moments beyond normalization are divergent. A similar calculation for a complex eigenvalue $z$ and the associated non-orthogonality factor in the complex Ginibre ensemble is presented as well and yields a distribution with the finite first moment. Its  'bulk' scaling limit yields a distribution whose first moment reproduces the well-known result of Chalker and Mehlig  \cite{ChalkerMehlig1998}, and we provide the 'edge' scaling distribution for this case as well. Our method involves evaluating the ensemble average of products and ratios of integer and half-integer powers of characteristic polynomials for Ginibre matrices, which we perform in the framework of a supersymmetry approach. Our paper complements recent studies by Bourgade and Dubach \cite{BourgadeDubach}.

\end{abstract}

\section{Introduction}
Let $\mathbf{x}$ be a $N-$ component column vector, real or complex.  We will use $\mathbf{x}^T=(x_1,\ldots, x_N)$ to denote the corresponding transposed row vector ( and similar notation for matrices), and $\mathbf{x}^*=(\overline{x}_1,\ldots, \overline{x}_N)$ for the Hermitian conjugate, with bar standing for complex conjugation. The inner product of two such vectors will be denoted  as $\mathbf{x}_1^*\mathbf{x}_2=\sum_{i=1}^N \overline{x}_{1i}x_{2i}$.

Let  $G$ be a $N\times N$ matrix which we assume to be non-selfadjoint and not normal : $G^*\ne G, \, G^{*}G\ne GG^{*}$. We will further assume that all $N$ eigenvalues $\lambda_a, \, a=1,\ldots,N$ of this matrix, which are in general complex numbers, have multiplicity one. Then the matrix is diagonalizable by a similarity transformation: $G=S\Lambda S^{-1}$ where $\Lambda=\mbox{diag}\left(\lambda_1,\ldots,\lambda_N\right)$ and $S$ is in general non-unitary: $S^*\ne S^{-1}$. The associated {\it right} eigenvectors defined by  $G\,\mathbf{x}_{Ra}=\lambda_a \mathbf{x}_{Ra}$ are columns of the matrix $S$, whereas their {\it left} counterparts
 satisfying  $\mathbf{x}_{La}^* G=\lambda_a \mathbf{x}^*_{La}$ form the rows of $S^{-1}$, and generically $\mathbf{x}_{Ra}\ne \mathbf{x}_{La}$. The sets $\mathbf{x}^*_{La}$ and $\mathbf{x}_{Ra}$ of left and right eigenvectors  can always be chosen to satisfy the {\it bi-orthonormality} condition $\mathbf{x}^*_{La}\mathbf{x}_{Rb}=\delta_{ab}$ for $a,b=1,\ldots,N$, but non-unitarity of $S$ implies that $\mathbf{x}^*_{Rb}\mathbf{x}_{Ra}\ne \delta_{ab}$ and similarly $\mathbf{x}^*_{La}\mathbf{x}_{Lb}\ne \delta_{ab}$.
  Then the simplest informative object characterizing the eigenvector non-orthogonality is the so-called 'overlap matrix' $\mathcal{O}_{ab}=(\mathbf{x}^*_{La}\mathbf{x}_{Lb})(\mathbf{x}^*_{Rb}\mathbf{x}_{Ra})$.
  In particular, the real diagonal entries $\mathcal{O}_{aa}$ are
  known in the literature on numerical analysis as {\it eigenvalue condition numbers} and characterize sensitivity of eigenvalues $\lambda_a$ to perturbation of entries of $G$, see e.g. \cite{TrefethenBook}. Namely, consider a family of matrices $G(\alpha)=G+\alpha V$, with $V$ being an arbitrary matrix whose $2-$norm is fixed as $||V||_2=1$, whereas $\alpha$ is a real parameter controlling the magnitude of the perturbation. Denote, for a given $V$, the eigenvalues of $G(\alpha)$ as $\lambda_a(\alpha)$
and consider $\dot{\lambda_a}(\alpha)=\frac{d\lambda_a}{d\alpha}$. A standard calculation using bi-orthonormality shows that $\dot{\lambda}_a(0)= \mathbf{x}^*_{La}V\mathbf{x}_{Ra}$ and therefore $\left|\dot{\lambda}_a(0)\right|\le  |\mathbf{x}_{La}| \, ||V||_2\, |\mathbf{x}_{Ra}|=\mathcal{O}_{aa}^{1/2}$ showing indeed that $\mathcal{O}_{aa}$ controls the speed of change of eigenvalues under perturbation.
As for some classes of non-normal matrices $\mathcal{O}_{aa}\gg 1$, their eigenvalues could be much more sensitive to perturbations in comparison with their normal counterparts.

  If the matrix $G$ is random, it makes sense to be interested in statistics of $\mathcal{O}_{ab}$. This line of research originated from the influential papers by Chalker and Mehlig \cite{ChalkerMehlig1998, MehligChalker2000} who were the first to evaluate asymptotically, for large $N\gg 1$, the lowest moments of the form
\begin{equation}\label{meanoverlap}
\mathcal{O}(z)=\left\langle \sum_{a=1}^N \mathcal{O}_{aa}\delta(z-\lambda_a)\right\rangle_{Gin_2}, \quad \mathcal{O}(z_1,z_2)=\left\langle \sum_{a\ne b}^N
\mathcal{O}_{ab}\delta(z_1-\lambda_a)\delta(z_2-\lambda_b)\right\rangle_{Gin_2}
\end{equation}
where $\delta(z-\lambda_a)$ stands for the appropriate Dirac delta-distribution (so that e.g. the empirical density of eigenvalues at a (in general, complex) point $z$ is given by $\sum_{a=1}^N \delta(z-\lambda_a)$). The brackets $\left \langle \ldots \right\rangle_{Gin_2}$ denote here the expectation with respect to the probability measure on $G$ known as the complex Ginibre ensemble, which we denote in this paper as  $Gin_2$ to reflect that ensembles with complex entries are usually characterized by the Dyson index $\beta=2$, see below. The probability measure on $G$ with real entries known as the real Ginibre ensemble will be denoted correspondingly with $Gin_1$.

For $\beta=2$ Chalker and Mehlig were able to extract the leading asymptotic behaviour of $\mathcal{O}(z)$ and $\mathcal{O}(z_1,z_2)$ in the $N\gg 1$ limit. In particular, they found that $\mathcal{O}(z)\approx N^2(1-|z|^2)$ inside the' Ginibre circle' characterized by the asymptotic mean eigenvalue density $\left\langle \sum_{a=1}^N \delta(z-\lambda_a)\right\rangle_{Gin_2}\approx \frac{N}{\pi^2}$ for $|z|^2<1$ and zero otherwise. This suggests that typically one should expect $\mathcal{O}_{aa}\sim N$ for eigenvalues inside the circle, which is parametrically larger than $\mathcal{O}_{aa}=1$ typical for normal matrices.

In the last decades there was steady growth of interest in understanding properties of non-orthogonal random eigenvectors in theoretical physics, see
\cite{Janik1999,Scho00,MehligSanter2001,FyoMel2002,GS2011,FyoSav2012,NonorthExper,Burda2014,Burda2015,Belinchi2017,Burda2017,NT2018}, with emphasis on calculating the Chalker-Mehlig correlators (\ref{meanoverlap}) and related objects beyond the framework of the complex Ginibre ensemble.  One motivation comes from the abovementioned relevance of eigenvector correlations for describing the motion of complex eigenvalues under perturbations of the ensemble, see e.g.  \cite{Movassagh}, and associated Dysonian dynamics, see e.g. \cite{Burda2015} and Appendix A of \cite{BourgadeDubach}. Note that the non-orthogonality factors reflect non-normality of the matrix, which in the context of dynamical systems is known to give rise to a long transient behaviour, see a general discussion in \cite{Trefethen1993,Somp2008}. In a related setting non-symmetric matrices appear very naturally via linearization around an equilibrium in a complicated nonlinear dynamical system \cite{May1972,FyoKhor2016}, and the non-orthogonality factors then control transients in a relaxation towards equilibrium \cite{Grela2017}. Non-orthogonality also plays some role in analysis of spectral outliers in non-selfadjoint matrices, see e.g. \cite{MetzNeri} and references therein.
 Another strong motivation comes from the field of quantum chaotic scattering, where non-selfadjoint random matrices of special type (different from the Ginibre ensembles) play a prominent role, see
e.g. \cite{FyoSom1997,FyoSomRev2003,Rotter,Fyod11,RK2017} for the background information.    The corresponding non-orthogonality overlap matrix
$\mathcal{O}_{ab}$  shows up in various scattering observables, such as e.g. decay laws \cite{Savi97}, 'Petermann factors'  describing excess noise in open laser resonators \cite{Scho00}, as well as in sensitivity of the resonance widths to small perturbations \cite{FyoSav2012,NonorthExper}.
Unfortunately, main progress in understanding properties of the bi-orthogonal eigenvectors for such ensembles relied on treating non-Hermiticity perturbatively in a small parameter, whereas non-perturbative results are scarce \cite{Scho00,MehligSanter2001,FyoMel2002}.

In the Mathematics community a systematic rigorous research in this direction seems to have started only recently
\cite{WaltersStarr2015}. In a very recent development Bourgade and Dubach \cite{BourgadeDubach} demonstrated a possibility to find
the law of the random variable  $\mathcal{O}_{aa}$ for the complex Ginibre ensemble, asymptotically for large $N$,  and provided a valuable information
about the off-diagonal correlations between the two different eigenvectors at various scales of eigenvalue separation (the so-called 'microscopic' vs. 'mesoscopic' scales). That work motivated the present paper, where we use a rather different approach to consider the following object
\begin{equation}\label{overlapdistr}
\mathcal{P}(t,z)=\left\langle \sum_a \delta\left(\mathcal{O}_{aa}-1-t\right)\delta(z-\lambda_a)\right\rangle_{Gin_{1} \, \mbox{\tiny or} \, Gin_{2} }.
\end{equation}
 interpreted as the (conditional) probability density of the 'diagonal' (or 'self-overlap') non-orthogonality  factor $t=O_{aa}-1$ for  the right and left eigenvectors corresponding to eigenvalues in the vicinity of a point $z=x+iy$ in the complex plane. We will call it for brevity the joint probability density (JPD) of the two variables, $t$ and $z$. Note that this JPD is normalized in such a way that the integral
 $\int_{\mathbb{R}_{+}} \mathcal{P}(t,z)\,dt=\left\langle \sum_a \delta(z-\lambda_a)\right\rangle:=\rho_N(z)$ coincides with the mean density $\rho_N(z)$
of eigenvalues around point $z$ in the complex plane.

 Naturally, the JPD function $\mathcal{P}(t,z)$ can be defined for a general random matrix $G$ and, in particular, may be used to quantify the statistics
of eigenvalue sensitivity parameters for such matrices. To give an example, consider again the family of matrices $G(\alpha)=G+\alpha V$, but choose the perturbation $V$ to be a random matrix independent of $G$.  For simplicity one may take $V$ to be proportional to a random complex Ginibre matrix, and normalized in such a way that its entries $V_{ij}$ are i.i.d. mean zero complex numbers  with the variance $\left\langle \overline{V}_{ij}V_{kl}\right\rangle_{Gin_2}=\frac{1}{N}\delta_{ik}\delta_{jl}$. Then the eigenvalue sensitivity to such a perturbation is given by $\dot{\lambda}_a(0)= \mathbf{x}^*_{La}V\mathbf{x}_{Ra}$ and for a fixed $G$ becomes
a complex Gaussian variable with mean zero and variance $\left\langle |\dot{\lambda}_a(0)|^2\right\rangle_V= \frac{1}{N}\mathcal{O}_{aa}$.
 Define now the probability density $\pi(w,z)$ of the eigenvalue sensitivity at a point $z$ of the complex plane via $\pi(w,z)=\left\langle\sum_a \delta\left(w-\dot{\lambda}_a(0)\right)\delta(z-\lambda_a)\right\rangle_{G,V}$ where the ensemble averaging goes both over $G$ and over $V$.
 Since the complex Gaussian variable $w=\dot{\lambda}_a(0)$ has the density $\pi(w)=\frac{1}{\pi \left\langle |\dot{\lambda}_a(0)|^2\right\rangle_V }e^{-|w|^2/\left\langle |\dot{\lambda}_a(0)|^2\right\rangle_V}$ with respect to the Lebesgue measure
 $d(Im\,w)d(Re\,w)=\frac{1}{2}dwd\overline{w}$ and recalling $\mathcal{O}_{aa}=1+t$ we immediately see that
 \begin{equation}\label{eigensens}
\pi(w,z)=\int_0^{\infty}\frac{N}{\pi(1+t)}e^{-\frac{N}{1+t}|w|^2}\mathcal{P}(t,z)\,dt\,,
\end{equation}
relating the statistics of the eigenvalue sensitivity in that case to the knowledge of $\mathcal{P}(t,z)$.

In this paper we concentrate on finding explicit expressions for $\mathcal{P}(t,z)$ for Ginibre matrices, both real and complex. We first consider in Section 3 the case $Gin_{1}$ of real Ginibre matrices with $\beta=1$.
To this end it is useful to recall that real-valued matrices may have
either purely real eigenvalues or  pairs of complex conjugate eigenvalues. As the result, $\rho_N(z)$ for real Ginibre ensemble necessarily has the form $\rho_N(z)=\rho^{(c)}_N(z)+\delta\left(y \right)\rho^{(r)}_N(x)$, where the non-singular part $\rho^{(c)}_N(z)$ describes the mean density of complex  eigenvalues, whereas $\rho^{(r)}_N(x)$ describes the mean density of purely real eigenvalues, so that $\int_{a}^{b} \rho^{(r)}_N(\lambda)d\lambda$ stands for the mean number of real eigenvalues in an interval $[a,b]$ of the real axis. As a consequence, the introduced JPD $\mathcal{P}(t,z)$ inherits the same structure  $\mathcal{P}(t,z)=\mathcal{P}^{(c)}(t,z)+\delta(y)\mathcal{P}^{(r)}(t,x)$.

The organization of the paper is as follows.  A summary of the main results and discussion of possible directions for the future work is presented in the Section 2. We start our consideration with demonstrating in Section 3 a way to evaluate $\mathcal{P}^{(r)}(t,\lambda)$, which describes non-orthogonality factor for eigenvectors associated with a real eigenvalue $\lambda$ of the real Ginibre ensemble. First, we reduce the problem of finding $\mathcal{P}^{(r)}(t,\lambda)$ to a problem of evaluating certain ratios of determinants of random real-symmetric matrices with block structure, which as one may eventually see are intimately related to a deformed version of the so-called real chiral ensemble. Technical calculations within a framework of the supersymmetry approach which proves to be an efficient technical tool for dealing with such ratios of determinants are presented in Section 3.3. Our approach yields exact and explicit formula for any size $N$, which is then amenable to extracting the appropriate 'bulk' and 'edge' scaling limits as $N\to \infty$. The problem of evaluating $\mathcal{P}^{(c)}(t,z)$ for real Ginibre matrices remains presently outstanding, and we hope to be able to address it in a future publication.

 In the next Section 4 we apply essentially the same method for evaluating $\mathcal{P}^{(c)}(t,z)$ in the complex Ginibre ensemble $Gin_{2}$, i.e. $\beta=2$. The computations and results become somewhat more technically involved, and considerably simplify only for the special case $|z|=0$. General case is treated again by the supersymmetry approach outlined in Section 4.4. Eventually, we present an explicit finite-$N$ expression for any $z$, and then extract the corresponding 'bulk' and 'edge' scaling limits.

{\bf Acknowledgements}.  The author is most grateful to Paul Bourgade and Guillaume Dubach for generously communicating their unpublished results at an early stage which stimulated his own research on the topic.  Ramis Movassagh is acknowledged for an interesting discussion and bringing reference \cite{Movassagh} to the author's attention,   Gernot Akemann for pointing out \cite{QCDfinT} and Peter Forrester for mentioning \cite{DesFor2008}. Jacek Grela and Eugene Strahov are acknowledged for their collaboration on the associated analysis of Eq.(\ref{invcharpol}) using different methods \cite{FyoGrelStrah2017}. The present paper was started when preparing a lecture course for PCMI Summer School 2017, and essentially completed during the Beg Rohu Summer School 2017.  The author would like to thank the organizers and participants of the schools for creating a stimulating atmosphere, and for the financial support of his participation in the events, in particular from the NSF grant DMS:1441467. The research at King's College London was supported by  EPSRC grant  EP/N009436/1 "The many faces of random characteristic polynomials".

\section{Discussion of the main results}
\subsection{Real Ginibre ensemble}
Note that the left and right eigenvectors of real-valued matrices corresponding to real eigenvalues $\lambda$ can be chosen real as well.
Hence we may write $\mathbf{x}_{\lambda,L}^T$ instead of $\mathbf{x}_{\lambda,L}^*$.
\begin{theorem}\label{Theorem1}
  Consider the real Ginibre ensemble of $N\times N$ square matrix $G \in \mathcal{M}_N\left(\R\right)$ with independent
identically distributed standard Gaussian real matrix elements:
	\begin{equation}\label{Ginelements}
		G_{j,k} \sim \mathcal{N}\left(0,1\right).
	\end{equation}
 Let $\lambda$ be a real eigenvalue of $G$ which without reducing generality may be assumed simple as Ginibre matrices  with multiple eigenvalues have zero Lebesgue measure in
	 $\R^{N\times N}$. Further denote $\mathbf{x}_{\lambda,R}$ and $\mathbf{x}_{\lambda,L}^T$ the associated right and left  eigenvectors chosen to satisfy the normalization condition $\mathbf{x}^T_{\lambda,L}\mathbf{x}_{\lambda,R}=1$. Define  the so-called 'diagonal' (or 'self-overlap') nonorthogonality factor ${\cal O}_{\lambda}=(\mathbf{x}^T_{\lambda,L}\mathbf{x}_{\lambda,L})(\mathbf{x}^T_{\lambda,R}\mathbf{x}_{\lambda,R})$.  Then the joint probability density $\mathcal{P}^{(r)}(t,\lambda)$ of the variables $t={\cal O}_{\lambda}-1$ and the eigenvalue $\lambda$ is given for any $N\ge 2$ by
 \begin{equation}\label{Mainpositive}
\mathcal{P}^{(r)}(t,\lambda)=\frac{1}{2\sqrt{2\pi}} e^{-\frac{\lambda^2}{2}\left(1+ \frac{t}{1+t}\right)}\frac{t^{\frac{N-3}{2}}}{(1+t)^{\frac{N+1}{2}}}
\sum_{k=0}^{N-1}\frac{\lambda^{2k}}{k!}\left[(N-1-k)+\frac{k}{1+t}\right]
\end{equation}

\end{theorem}
In what follows we will frequently omit the index $\lambda$ in eigenvectors to lighten the notations, simply writing $\mathbf{x}^T_{L}$ or $\mathbf{x}_{R}$.

\begin{remark}
The above expression is a generalization of the exact mean density of purely real eigenvalues $ \rho^{(r)}_N(\lambda)$ for real Ginibre matrices of size $N$ explicit expression for which is known due to Edelman, Kostlan and Schub \cite{EKS}, see also \cite{Ed:97}:
 \begin{equation}\label{Edelman}
\rho^{(r)}_N(\lambda)=\frac{1}{\sqrt{2\pi}(N-2)!}\left[\Gamma\left(N-1,\lambda^2\right)+|\lambda|^{N-1}e^{-\frac{\lambda^2}{2}}\int_0^{|\lambda|} e^{-\frac{u^2}{2}} u^{N-2}\,du\right]\,.
\end{equation}
where
\begin{equation}\label{Gamma}
\Gamma\left(N,a\right)=(N-1)!e^{-a}\sum_{k=0}^{N-1} a^k/k!=\int_a^{\infty}e^{-u}u^{N-1}\,du\,
\end{equation}
is the incomplete $\Gamma$-function.
The expression (\ref{Edelman}) can be most easily recovered from (\ref{Mainpositive}) as $\rho^{(r)}_N(\lambda)=\int_{\mathbb{R}_{+}} \mathcal{P}(t,\lambda)\,dt$ after
rewriting (\ref{Mainpositive}) in an equivalent form as
 \begin{equation}\label{Main}
\mathcal{P}^{(r)}(t,\lambda)=\frac{1}{2\sqrt{2\pi}} e^{\frac{\lambda^2}{2} \frac{1}{1+t}}\frac{1}{t(1+t)}\left(\frac{t}{1+t}\right)^{\frac{N-1}{2}}\frac{1}{(N-2)!}
\left[\Gamma\left(N,\lambda^2\right)-\lambda^2\frac{t}{1+t}\Gamma\left(N-1,\lambda^2\right)\right]
\end{equation}
\[
=\frac{1}{2\sqrt{2\pi}} e^{\frac{\lambda^2}{2} \frac{1}{1+t}}\frac{1}{t(1+t)}\left(\frac{t}{1+t}\right)^{\frac{N-1}{2}}
\left[\frac{e^{-\lambda^2}\lambda^{2(N-1)}}{(N-2)!}+\frac{1}{(N-2)!}\Gamma\left(N-1,\lambda^2\right)\left((N-1)-\lambda^2\frac{t}{1+t}\right)\right]
\]
and then introducing $u=|\lambda|\sqrt{\frac{t}{1+t}}$
as the integration variable, cf. (\ref{EKS}) below.
  \end{remark}
 The expressions (\ref{Mainpositive}) or (\ref{Main}) show that the random 'overlap' variable $t$ is maximally heavy-tailed, so that all positive integer moments $\mathbb{E}[t^\mu], \,\mu\ge 1$ for a fixed $\lambda$ are divergent due to the 'fat tail' $\mathcal{P}(t,\lambda)\sim t^{-2}$ as $ t\to \infty$, and only the normalization $\mu=0$ is finite and yields the density (\ref{Edelman}).

Being exact, the expression (\ref{Mainpositive}) can be further analyzed in interesting scaling limits as $N\to \infty$. In fact, we find the form (\ref{Main}) most suitable for such an analysis. In particular, by rescaling $\lambda=\sqrt{N} x, \,\, t=N s$ (which is standard to call the {\it bulk scaling} limit), then considering $x,s$ as fixed when $N\to \infty$ and exploiting the appropriate asymptotic behaviour of the incomplete $\Gamma$-function:
\begin{equation}\label{Gammalim}
\lim_{N\to\infty}\frac{\Gamma\left(N-1,Na\right)}{(N-2)!}=1, \, \mbox{if}\,\, 0<a<1\,\, \mbox{and}\,\, 0 \,\, \mbox{if}\,\, a>1
\end{equation}
one easily finds that $\lim_{N\to \infty} N\mathcal{P}^{(r)}(N s\,, \sqrt{N} x )=\mathcal{P}^{(r)}_{bulk}(s,x)$ where
\begin{equation}\label{Mainscalingbulk}
\mathcal{P}^{(r)}_{bulk}(s,x)=\frac{1}{2\sqrt{2\pi}} \frac{e^{-\frac{(1-x^2)}{2s}}}{s^2} (1-x^2), \quad |x|<1
\end{equation}
and $\mathcal{P}^{(r)}_{bulk}(s,x)=0$ otherwise. All positive integer moments $\mathbb{E}[s^\mu], \,\mu\ge 1$ are divergent as before, whereas for $\mu=0$ we have
\begin{equation}\label{Mainscalingbulk}
\rho^{(r)}_{bulk}(x)=\int_{\mathbb{R}_{+}}\mathcal{P}^{(r)}_{bulk}(s,x)\,ds=\frac{1}{\sqrt{2\pi}}, \quad |x|<1\quad
\end{equation}
and $\rho^{(r)}_{bulk}(x)=0$ otherwise, in full agreement with $\rho^{(r)}_{bulk}(x)$ being the limiting mean density of real eigenvalues within the bulk of the spectrum of the real Ginibre ensemble, which is known to be uniform inside its support.

Another natural {\it edge scaling} limit arises in the vicinity of the edge of the support of limiting spectral measure for real eigenvalues,
that is for $\lambda=\sqrt{N}+\delta$, with $\delta<\infty$ being fixed. It is easy to understand that the variable $t$ needs to be rescaled in this
regime as $t=\sqrt{N}\sigma$, keeping $\sigma$ fixed. A straightforward calculation using the well-known asymptotics
\begin{equation}\label{Gammaedgelim}
\lim_{N\to\infty}\frac{\Gamma\left(N-1,N\left(1+2\delta N^{-1/2}\right)\right)}{(N-2)!}=\frac{1}{\sqrt{2\pi}}\int_{2\delta}^{\infty}e^{-\frac{v^2}{2}}\,dv
\end{equation}
then yields    $\lim_{N\to \infty} \sqrt{N}\mathcal{P}(\sqrt{N} \sigma, \, \, \sqrt{N} +\delta)=\mathcal{P}^{(r)}_{edge}(\sigma,\delta)$ where
\begin{equation}\label{Mainscalingbulk}
\mathcal{P}^{(r)}_{edge}(\sigma,\delta)=\frac{1}{2\sqrt{2\pi}} \frac{1}{\sigma^2}
e^{-\frac{1}{4\sigma^2}+\frac{\delta}{\sigma}}
 \left[\frac{1}{\sqrt{2\pi}}e^{-2\delta^2}+\left(\frac{1}{\sigma}-2\delta\right)\frac{1}{\sqrt{2\pi}}\int_{2\delta}^{\infty}e^{-\frac{v^2}{2}}\,dv\right].
\end{equation}
In particular, integrating the above over $\sigma$ gives
\begin{equation}\label{Mainscalingbulk}
\rho^{(r)}_{edge}(\delta)=\int_{\mathbb{R}_{+}}\mathcal{P}^{(r)}_{edge}(\sigma,\delta)\,d\sigma
=\frac{1}{2\sqrt{2\pi}}\left[1-\mbox{erf}(\delta\sqrt{2})+\frac{1}{\sqrt{2}}e^{-\delta^2}\left(1+\mbox{erf}(\delta)\right)\right]\,,
\end{equation}
where  $\mbox{erf}(\delta)=\frac{2}{\sqrt{\pi}}\int_0^{\delta}e^{-u^2}\,du$. This expression is
in full agreement with one for the limiting mean density of real eigenvalues at the edge of the spectrum of the real Ginibre ensemble, see
\cite{ForNag2008}.
\subsection{Complex Ginibre ensemble}
For the case of complex Ginibre ensemble $Gin_2$ the corresponding joint probability density $\mathcal{P}^{(c)}(t,z)$ of the non-orthogonality variable $t={\cal O}_{z}-1$ and the associated complex eigenvalue $z$  can be found in explicit form for finite $N$ as well, but turns out to be given by a much more cumbersome expression in comparison with the real Ginibre case.
\begin{theorem}\label{Theorem2}
  Consider the complex Ginibre ensemble of $N\times N$ square matrix $G \in \mathcal{M}_N\left(\C\right)$ with independent
identically distributed complex Gaussian matrix elements
	\begin{equation}\label{Ginelements}
		G_{j,k} = g^{\left(1\right)}_{j,k}+
			i\cdot g^{\left(2\right)}_{j,k},
			\quad
			\mbox{with i.i.d.}
			\,\,
			g^{\left(\cdot\right)}_{j,k}
			\sim \mathcal{N}\left(0,1/2\right),
	\end{equation}
 Let $z$ be a complex eigenvalue of $G$. Then the joint probability density $\mathcal{P}^{(c)}(t,z)$ of the self-overlap non-orthogonality variable $t={\cal O}_{z}-1$ and the associated complex eigenvalue $z$  for $N\ge 2$ is given by
\begin{equation}\label{MainComp}
\mathcal{P}^{(c)}(t,z)=\frac{1}{\pi (N-1)!(N-2)!}\,\frac{e^{|z|^2\frac{1}{1+t}}}{(1+t)^3}
\left(\frac{t}{1+t}\right)^{N-2}\left(D_1^{(N)}+|z|^2\frac{D_2^{(N)}}{1+t}+|z|^4\frac{d_1^{(N)}}{(1+t)^2}\right)
\end{equation}
where  $D^{(N)}_1$ and $D^{(N)}_2$ are defined as
\begin{equation}\label{LaplacecompFINALcoeffC}
D_1^{(N)}=|z|^4 (N-1)(N-2)d_1^{(N-1)}+\left[(N-1)N-2|z|^2(N +  |z|^2)\right]d_1^{(N)}
\end{equation}
\[
-|z|^2(N-2)(N-|z|^2)d_2^{(N-1)}+|z|^2d_2^{(N)}\,,
\]
\begin{equation}\label{LaplacecompFINALcoeffD}
\quad D_2^{(N)}=2N\,d_1^{(N)}-|z|^2(N-2)d_2^{(N-1)}
\end{equation}
 and $d_1^{(N)}$, $d_2^{(N)}$ are functions of $|z|^2$ explicitly defined via the relations
to the incomplete $\Gamma-$function as
\begin{equation}\label{LaplacecompFINALcoeffA}
d_1^{(N)}=\Gamma\left(N-1,|z|^2\right)\Gamma\left(N+1,|z|^2\right)-\Gamma\left(N,|z|^2\right)\Gamma\left(N,|z|^2\right)
\end{equation}
\begin{equation}\label{LaplacecompFINALcoeffB}
d_2^{(N)}=\Gamma\left(N-1,|z|^2\right)\Gamma\left(N+2,|z|^2\right)-\Gamma\left(N,|z|^2\right)\Gamma\left(N+1,|z|^2\right)
\end{equation}
Note that $d^{(N)}_{1},d^{(N)}_{2}$ and $D^{(N)}_{1},D^{(N)}_{2}$ depend only on $|z|^2$ but not on the variable $t$.
\end{theorem}
\begin{remark}
The expression (\ref{MainComp}) is a generalization of the well-known mean density of eigenvalues of the complex Ginibre ensemble, see e.g. \cite{KS2011} :
 \begin{equation}\label{denscomp}
\rho_N^{(c)}(z)=\frac{1}{\pi} e^{-|z|^2}\sum_{n=0}^{N-1}\frac{|z|^{2n}}{n!}
\end{equation}
The latter can indeed be recovered as $\rho^{(c)}_N(z)=\int_{\mathbb{R}_{+}} \mathcal{P}^{(c)}(t,z)\,dt$ as is possible to check employing Mathematica.
  \end{remark}

The joint density $\mathcal{P}^{(c)}(t,z)$ at fixed $z$ decays at large arguments $t\gg 1$ as  $\mathcal{P}^{(c)}(t,z)\sim t^{-3}$ as was already anticipated by Mehlig and Chalker on the basis of $N=2$ example and informal eigenvalue repulsion arguments \cite{MehligChalker2000}. In contrast to the real Ginibre case such density does have the finite first moment.
Only for the  special value $z=0$ the above joint density significantly simplifies and is given by
 \begin{equation}\label{MainCompzero}
\mathcal{P}^{(c)}(t,0)=\frac{N(N-1)}{\pi} \frac{t^{N-2}}{(1+t)^{N+1}}, \quad N\ge 2
\end{equation}

Despite the relative complexity of (\ref{MainComp}), its bulk rescaling limit $ z=\sqrt{N} w, \, t=N s$,  with $w,s$ being fixed when $N\to \infty$, can be straightforwardly extracted. To this end, it is convenient to use the following integral representations for $d_1^{(N)}$ and $d_2^{(N)}$ (following from
combining (\ref{LaplacecompFINALcoeffA}) and  (\ref{LaplacecompFINALcoeffB}) with  (\ref{Gamma}) and appropriate rescaling):
\begin{equation}\label{d1int}
d_1^{(N)}=\frac{1}{2}N^{2N}\int_{|z|^2/N}^{\infty} du_1\int_{|z|^2/N}^{\infty} du_2 \, (u_1-u_2)^2\, \frac{1}{(u_1u_2)^2}\,e^{-N(u_1-\ln{u_1}+u_2-\ln{u_2})}
\end{equation}
and
\begin{equation}\label{d1int}
d_2^{(N)}=\frac{1}{2}N^{2N+1}\int_{|z|^2/N}^{\infty} du_1\int_{|z|^2/N}^{\infty} du_2 \, (u_1-u_2)^2 \, \frac{(u_1+u_2)}{(u_1u_2)^2}\,e^{-N(u_1-\ln{u_1}+u_2-\ln{u_2})}
\end{equation}
which for large $N\gg 1$ are easily amenable to the standard asymptotic analysis by the Laplace method. In this way we find in the 'bulk' scaling limit
  for $|w|^2<1$  the following $w-$ independent asymptotic behaviour:
\begin{equation}\label{asycomp1}
d_1^{(N)}\sim 2\pi N^{2N-2}e^{-2N}\sim \frac{1}{N^3}(N!)^2,  \quad d_2^{(N)}\sim 2N\,d_1^{(N)}
\end{equation}
where $a_N\sim b_N$ means $\lim_{N\to \infty} a_N/b_N=1$. This implies $ d_1^{(N-1)}\sim \frac{1}{N^2}\,d_1^{(N)}, \quad d_2^{(N-1)}\sim \frac{2}{N}\,d_1^{(N)}$ and then via (\ref{LaplacecompFINALcoeffC}) and (\ref{LaplacecompFINALcoeffD}) we further find
\begin{equation}\label{asycomp2}
D_1^{(N)}\sim  \frac{1}{N} \, (N!)^2 (1-|w|^2)^2, \quad  D_2^{(N)}\sim \frac{2}{N^2} (N!)^2 (1-|w|^2)
\end{equation}
We then see that in the 'bulk' limit the first term $D^{(N)}_1$ in the brackets of (\ref{MainComp}) is dominant in comparison with the other two, and taking the limit
 $\lim_{N\to \infty} N\mathcal{P}^{(c)}(t=N s,z=\sqrt{N} w)=\mathcal{P}^{(c)}_{bulk}(s,w)$ one finds
\begin{equation}\label{MainscalingbulkComp}
\mathcal{P}^{(c)}_{bulk}(s,w)= \frac{(1-|w|^2)^2}{\pi s^3}\,e^{-\frac{1-|w|^2}{s}}, \quad |w|<1
\end{equation}
and zero otherwise. This expression agrees with results obtained by P. Bourgade and G. Dubach \cite{BourgadeDubach} in a different approach to the problem. Its first moment is precisely $\frac{1}{\pi}(1-|w|^2)$ inside the bulk of the spectrum, in agreement with the expression by Chalker and Mehlig.

Finally, one also can extract the corresponding edge asymptotics by replacing $|z| =\sqrt{N}+\delta$ and $t=\sqrt{N}\sigma$ and performing the limit $N\to \infty$. With a help of the Mathematica package\footnote{The author is grateful to J. Grela for his help with utilizing Wolfram Mathematica for that purpose.} one then finds that
\[\lim_{N\to\infty} N^{3/2}\mathcal{P}^{(c)}(\sqrt{N}\sigma, \sqrt{N}+\delta)=\mathcal{P}^{(c)}_{edge}(\sigma,\delta)\] where

\begin{equation}\label{edgecomp}
\mathcal{P}^{(c)}_{edge}(\sigma,\delta)=\frac{1}{2\pi\sigma^5}\,e^{-\frac{\Delta^2}{2\sigma^2}}
\left\{\frac{e^{-2\delta^2}}{\pi}\left(2\sigma^2-\Delta\right)-\frac{1}{\sqrt{2\pi}}\left(4\delta\sigma^2-\Delta(2\delta+\sigma)\right)
\erfc\left({\sqrt{2}\delta}\right)\right.
\end{equation}
\[
\left. +\frac{e^{2\delta^2}}{2}\left(\Delta^2-\sigma^2\right)\erfc^2\left({\sqrt{2}\delta}\right)\right\}
\]
and we denoted $\Delta=1-2\sigma\delta$. In particular, one can check that integrating over $\sigma$ yields a well-known formula for the
mean edge density of complex eigenvalues:
\[
\rho_{edge}(\delta)=\int_{\mathbb{R}_{+}}\mathcal{P}^{(c)}_{edge}(\sigma,\delta)\,d\sigma=\frac{1}{2\pi}\erfc\left({\sqrt{2}\delta}\right)
\]
One also can see that the 'bulk'(\ref{MainscalingbulkComp}) and 'edge' (\ref{edgecomp}) asymptotics match by replacing in the latter $\delta=\frac{1}{2}\sqrt{N}(|w|^2-1)$ and $\sigma=\sqrt{N}s$ and letting $N\to \infty$ for fixed $|w|<1$ and $s$, checking that
\[\lim_{N\to\infty} N^{1/2}\mathcal{P}^{(c)}_{edge}\left(\sigma=\sqrt{N}s, \delta=\frac{1}{2}\sqrt{N}(|w|^2-1)\right)=\mathcal{P}^{(c)}_{bulk}(s,w)\]

\subsection{Discussion of the method and open problems.}
Our approach consists of two steps. In the first step we show that the partial Schur decomposition of Ginibre matrices employed in works
\cite{EKS} and \cite{Ed:97} allows one to represent the JPD's $\mathcal{P}^{(r)}(t,\lambda)$ and $\mathcal{P}^{(c)}(t,z)$,  Laplace-transformed with respect to
the variable $t$, in terms of the following object:
\begin{equation}\label{DetGenbeta}
\mathcal{D}^{(L)}_{N,\beta}(z,p)=\left\langle \frac{\det^{\beta L/2}\left(z\,I_N-G\right)\left(\overline{z}\,I_N-G^*\right)}{\det^{\beta/2}
\left[\frac{2}{\beta} p\,I_N+\left(z\,I_N-G\right)\left(\overline{z}\,I_N-G^*\right)\right]} \right\rangle_{Gin_\beta}
\end{equation}
where $L=0,1,2,\ldots$ is an integer, $p\ge 0$, the parameter $\beta=2$ stands for the complex Ginibre ensemble and  $\beta=1$ for the real Ginibre one (in the latter case $z=\lambda$ is real), and $I_N$ standing for the $N\times N$ identity matrix. In fact, the goals of the present paper require evaluation of (\ref{DetGenbeta})
only for $L=2$, but it is interesting to consider a more general problem, see below.

 Note that for $\beta=2$ we deal here with expectation values involving integer powers of characteristic polynomials for non-selfadjoint matrices $G$ in both numerator and denominator. Studying similar objects for self-adjoint random matrices has a long history, see e.g. \cite{FyoStra2003, BorStra2006} for a background discussion and further references. At the same time, for $\beta=1$ a half-integer power in the denominator is involved.
 To deal with the latter challenge we employ one of very few techniques available in that case, the so-called supersymmetry approach, see \cite{ZirnSUSY,GuhrSUSY} for concise introductions and also \cite{FyoNock2015,Wirtzetal2015} for earlier computations involving half-integer powers of characteristic polynomials for real symmetric Gaussian random matrices. We find it convenient to use a (rigorous) variant of the approach proposed originally in \cite{YF2002} and the final expression for $\mathcal{D}_{N,1}^{(2)}(\lambda,p)$ is given in (\ref{DNequal}) or (\ref{DNFin}). As a by-product of the same calculation one also finds for $L=0$:
\begin{equation}\label{DNFink0beta1}
\mathcal{D}_{N,2}^{(L=0)}(\lambda,p) = \frac{1}{2^{N/2}\Gamma(N/2)} \int_{\mathbb{R}_+}\frac{dt}{t}e^{-pt} e^{-\frac{\lambda^2}{2\left(1+\frac{1}{t}\right)}}\left(\frac{t}{1+t}\right)^{N/2}
\end{equation}
 The same procedure works, with due modifications, for the complex case $\beta=2$, with the computational challenge now coming not from the half-integer power in the denominator, but from the higher integer power of the determinant in the numerator.  The actual calculation is very straighforward for $L=0$, becomes slightly more involved for $L=1$, and for the case of actual interest $L=2$ produces much more cumbersome expressions, see Sec. 4.1.2 for the derivation.  In the end we have to resort  to symbolic computer manipulations to deal with the ensuing integrals.  Here we simply quote the results for $L=0$ and $L=1$ for the sake of completeness:
\begin{equation}\label{DNFink0beta2}
\mathcal{D}_{N,2}^{(L=0)}(z,p) = \frac{1}{(N-1)!} \int_{\mathbb{R}_+}\frac{dt}{t}e^{-pt} e^{-\frac{|z|^2}{\left(1+\frac{1}{t}\right)}}\left(\frac{t}{1+t}\right)^{N}
\end{equation}
and \footnote {In fact the $L=1$ case was considered by a different variant of the supersymmetry approach in \cite{GrelaGuhr}, though the result was not presented in the form \ref{DNFink1beta2}.}
\begin{equation}\label{DNFink1beta2}
\mathcal{D}_{N,2}^{(L=1)}(z,p) = \frac{1}{(N-1)!} e^{|z|^2}\int_{\mathbb{R}_+}\frac{dt}{t(1+t)}e^{-pt} e^{-\frac{|z|^2}{\left(1+\frac{1}{t}\right)}}\left(\frac{t}{1+t}\right)^{N}\left[\Gamma(N+1,|z|^2)-\frac{t}{1+t}|z|^2 \Gamma(N,|z|^2)\right]
\end{equation}
In particular, by a direct integration one can check that  $\mathcal{D}_{N}^{(L=1)}(z,p=0)=1$, in agreement with the definition  (\ref{DetGenbeta}).

Given the complexity of arising expressions for $\beta=2$, it is worth to give a different perspective on the problem. To that end we note that by introducing the matrices $W=z-G$ our main object for $\beta=2$ case, namely $ \mathcal{D}_{N,2}^{(L)}(z,p)$, can be formally rewritten as
\begin{equation} \label{invcharpol}
 \mathcal{D}_{N,2}^{(L)}(z,p)= \int \, \frac{1}{\det{\left[p I_N +WW^*\right]}} \, P_{L,z}\left(W,W^*\right) dW\,dW^* ,
\end{equation}
with integration going over complex $N\times N$ matrices $W$ with the weight function $P_{L,z}\left(W,W^*\right)$ depending on an integer parameter $L=0,1,2,\ldots$, and on the complex parameter $z$:
 \begin{equation}\label{Ensemble}
		P_{L,z}\left(W,W^*\right) =e^{-N|z|^2} \left(\frac{1}{2\pi}\right)^{N^2}
		e^{-\mbox{\small Tr} \left( W W^*\right)+z\mbox{\small Tr}W^*+\overline{z}\mbox{\small Tr}W}\, \left[\det{WW^*}\right]^L
	\end{equation}
 The right-hand side of (\ref{invcharpol}) can be obviously interpreted as the mean inverse characteristic polynomial of the matrix $W^*W$ averaged over this 'ensemble' \footnote{Formally the weight defined in (\ref{Ensemble}) is not a probability measure for any $L\ne 0$ as it is not normalized to unity, but
  we disregard such difference for our goals.} closely related (though not identical for $L>0$) to a limiting case of versions of the chiral ensemble with a 'source' considered in \cite{QCDfinT} and \cite{DesFor2008}. Namely, let us consider a  more general version of (\ref{Ensemble}):
\begin{equation}\label{EnsembleGen}
		P^{(A)}_{L,z}\left(W,W^*\right) = C_{N,L} e^{-\sum_{i=1}^N a_i} \left(\frac{1}{2\pi}\right)^{N^2}
		e^{-\mbox{\small Tr} \left( W W^*\right)+\mbox{\small Tr}\left(W^*A\right)+\mbox{\small Tr}\left(A^*W\right)}\, \left[\det{WW^*}\right]^L
	\end{equation}
where the 'source' matrix $A$ is a fixed $N\times N$ complex matrix with the singular values (i.e. eigenvalues of $A^*A$) being (in general, distinct) non-negative real numbers $a_1,a_2,\ldots,a_N$. Obviously our previous choice corresponded to all $a_i$ equal to $|z|^2$. Note that the $n-$point correlation functions of eigenvalue densities for such type of a chiral ensemble (with a Hermitean source $A$) were derived in \cite{QCDfinT}, but their knowledge is not sufficient for our purposes. Some information for the mean inverse characteristic polynomial for the chiral ensemble with a 'source' similar to (\ref{EnsembleGen}) was given in the framework of the method of multiple orthogonal polynomials in \cite{DesFor2008}. In a separate paper \cite{FyoGrelStrah2017} we are providing the full analysis of the problem for $\beta=2$ for any integer positive $L$ and $N$ by deriving  the following
representation (see Proposition 3.9 in \cite{FyoGrelStrah2017} ):
\begin{equation}\label{Gen2mynorm}
 \mathcal{D}^{(L)}_{N,\beta=2}(z,p)=\int_0^{\infty} \frac{dt}{(1+t)^{L+1}} e^{-tp}\, \mathcal{G}^{(L)}_N\left(|z|^2,\frac{t}{1+t}\right)
\end{equation}
where we defined the following function of $\rho=|z|^2$ and $\tau=t/(1+t)$:
\begin{equation}\label{G1amynorm}
\mathcal{G}^{(L)}_N(\rho,\tau)=\frac{(-1)^N L!}{(N-1)!}\int_0^{\infty}dt_1\ldots \int_0^{\infty}dt_L \Delta^2(t_1,\ldots,t_L) \prod_{k=1}^L (t_k+\rho)^N\,e^{-t_k}
\end{equation}
\[
\times\frac{d^{N-1}}{d^{N-1}\rho}\left[ e^{-\rho\,\tau}\,{\bf L}_L\left(-\rho(1-\tau)\right)\frac{1}{\prod_{k=1}^L (t_k+\rho)}\right]
\]
where ${\bf L}_k(x)$ are Laguerre polynomials. The equivalence with (\ref{DNFink0beta2}) - (\ref{DNFink1beta2}) for $L=0,1$ can be straightforwardly verified (see the Appendix A of \cite{FyoGrelStrah2017}).  One can further perform the asymptotic analysis of (\ref{Gen2mynorm}) - (\ref{G1amynorm}) for $N\gg 1$ and extract the bulk scaling asymptotics relevant for the present paper in a more transparent and systematic way than is provided
by the supersymmetry approach in $\beta=2$ case. Nevertheless, supersymmetric treatment has its own merits: the method is robust and is expected to be generalizable to more general ensembles of non-selfadjoint random matrices lacking the full invariance of the Ginibre ensembles.

 The approach suggested in the present paper can be certainly adjusted for addressing overlaps of left/right eigenvectors corresponding to complex eigenvalues of real Ginibre ensemble, although in this way one encounters a few challenging technical problem not yet fully resolved. One can also envisage extensions addressing overlaps of two different eigenvectors, as well as posing similar questions for other types of non-Hermitian matrices, including those with quaternion structure for $\beta=4$, those relevant in the theory of chaotic scattering and those relevant in the Quantum Chromodynamics context\cite{Osborn}. We hope to be able to answer some of these questions in future publications.

\section{Proof of Theorem \ref{Theorem1}}

\subsection{Eigenvectors of real matrices via partial Schur decomposition}

Let $\lambda$ be a real eigenvalue of a matrix $G^{(N)}$ with real entries, and denote the associated real left and right eigenvectors as $\mathbf{x}^T_{L}$ and  $\mathbf{x}_{R}$.  Then, as is well-known, see e.g. \cite{EKS,Ed:97},  it is always possible to  represent the matrix $G^{(N)}$ as
\begin{equation}\label{reductreal}
G^{(N)}=P
\begin{pmatrix}
  \lambda  & {\bf w}^T \\
  0 & G^{(N-1)}
\end{pmatrix} P^T
\end{equation}
for some real $(N-1)$- component column vector $\mathbf{w}$ and a real matrix $ G^{(N-1)}$ of size  $(N-1)\times (N-1)$, whereas the matrix $P$ is real symmetric and orthogonal: $P^T=P$ and $P^2=I_N$. As the left/right eigenvectors of the matrix $\tilde{G}^{(N)}=\begin{pmatrix}
  \lambda  & {\bf w}^T \\
  0 & G^{(N-1)}
\end{pmatrix}$ are given in terms of $\mathbf{x}^T_{L}$ and $\mathbf{x}_{R}$ by $\mathbf{\tilde{x}}^T_L=\mathbf{x}^T_LP$ and   $\mathbf{\tilde{x}}_R=P\mathbf{x}_R$, we see that the corresponding self-overlaps remain invariant: ${\cal O}_{\lambda}={\cal \tilde{O}}_{\lambda}$. Hence one can use $\mathbf{\tilde{x}}^T_L$ and   $\mathbf{\tilde{x}}_R$ for our calculation. Moreover, it is immediately clear from the form of $\tilde{G}^{(N)}$ that $\mathbf{\tilde{x}}_R=\mathbf{e}_1:=(1,0,\ldots,0)^T$, whereas $\mathbf{\tilde{x}}^T_L=(1,\mathbf{b}_{\lambda}^T)$, where $\mathbf{b}_{\lambda}^T$ is an $(N-1)-$component real vector. This implies ${\cal \tilde{O}}_{\lambda}=1+\mathbf{b}_{\lambda}^T\mathbf{b}_{\lambda}$, reducing the study of the self-overlap to the statistics of the norm $\mathbf{b}_{\lambda}^T\mathbf{b}_{\lambda}$.
To this end we find that the vector $\mathbf{b}_{\lambda}^T$ can be readily expressed in terms of the resolvent of $\tilde{G}^{(N-1)}$ as
\begin{equation}\label{bvec}
\mathbf{w}^T+\mathbf{b}_{\lambda}^TG^{(N-1)}=\lambda\mathbf{b}_{\lambda}^T, \quad \Leftrightarrow \quad  \mathbf{b}_{\lambda}^T= \mathbf{w}^T\left(\lambda\, I_{N-1}-G^{(N-1)}\right)^{-1}
\end{equation}
implying
\begin{equation}\label{bvecnormreal}
 \mathbf{b}_{\lambda}^T\mathbf{b}_{\lambda}=\mathbf{w}^T\left[\left(\lambda\, I_{N-1}-G^{(N-1)}\right)^{T}\left(\lambda\, I_{N-1}-G^{(N-1)}\right)\right]^{-1}\,\mathbf{w}
\end{equation}
\subsection{ Partial Schur decomposition of the Real Ginibre Ensemble and overlap statistics. }

In this section we show how to reduce the calculation of the Laplace transform $\mathcal{L}_{\lambda}(p):=\int_0^{\infty}e^{-pt}\mathcal{P}^{(r)}(t,\lambda)\,dt$  of the JPD  $\mathcal{P}^{(r)}(t,\lambda)$  defined in (\ref{overlapdistr}) to evaluating the ensemble average for the ratio of certain determinants, see (\ref{Laplacereal}-\ref{Laplace1}).

   For the ensemble of $N\times N$ square matrices $G^{(N)}\equiv G \in \mathcal{M}_N\left(\R\right)$ with independent
identically distributed real matrix element $G_{j,k} \sim \mathcal{N}\left(0,1\right)$ we use the notation $Gin_1$ and denote by the angular brackets $\left\langle \cdots \right\rangle_{Gin_1}$
the expectation of any  function $F:\R^{N\times N} \to \C$ with respect to the associated probability distribution, though we will frequently omit the
corresponding subscript to lighten the formulas.
\begin{remark}
	In a similar way one defines the complex Ginibre ensemble $Gin_2$
		\begin{equation*}
			G_{j,k} = g^{\left(1\right)}_{j,k}+
			i\cdot g^{\left(2\right)}_{j,k},
			\quad
			\mbox{with i.i.d.}
			\,\,
			g^{\left(\cdot\right)}_{j,k}
			\sim \mathcal{N}\left(0,1/2\right),
		\end{equation*}
as well as the so-called quaternion $Gin_4$ ensemble which is however not considered in the present work.
\end{remark}
  Assigning the Dyson's index $\beta = 1,2,4$ \footnote{ In the literature one frequently uses the notation $GinOE$ emphasizing an orthogonal symmetry of the distribution, and correspondingly $GinUE$ and $GinSE$ for complex and quaternion real versions of Ginibre ensemble with $\beta=2$ and $\beta=4$, correspondingly.} one can write
for all three ensembles the Joint Probability Density (JPD) with respect to the flat Lebesgue measure  in
the form
	\begin{equation}\label{Ginibre}
		P_{\beta}\left(G\right) =
		\left(\frac{\beta}{2\pi}\right)^{\frac{\beta}{2} N^2}
		\exp{-\frac{\beta}{2}\Tr G G^*}.
	\end{equation}
where $G^*=\overline{G}^T$ stands for the Hermitian conjugation, and the bar for the complex conjugation.
In this section we will concentrate in detail on the real case $\beta=1$; similar treatment of  $\beta=2$ case will be briefly described in the last section.

We start with exploiting the following
\begin{proposition}[see Lemma 3.1, Lemma 3.2 and Theorem 3.1 in \cite{EKS}]\label{t:jpdfreal}

Assume that the JPD of $G$ is given by (\ref{Ginibre}) with $\beta=1$, and apply the transformation (\ref{reductreal}). Then the joint probability density
of elements of the matrix  $\tilde{G}$ is given by

\begin{equation}\label{volchange}
{\cal P}(\tilde{G})d\tilde{G}={\cal C}_{1,N}|\det{(\lambda I_{N-1}-G^{(N-1)})}|\,e^{-\frac{1}{2}\left(\lambda^2+\mathbf{w}^T\mathbf{w}+\Tr G^{(N-1)}G^{(N-1)T}\right)}\,d\lambda\,d\mathbf{w}\,dG^{(N-1)} \, .
\end{equation}
with some normalization constant ${\cal C}_{1,N}$.
\end{proposition}
The above j.p.d. can be used to calculate the Laplace transform of the probability density for our main object of interest, the random variable
$\mathbf{b}_{\lambda}^T\mathbf{b}_{\lambda}$ which for a given value of $\lambda$ is given by (\ref{bvecnormreal}), or equivalently the characteristic function
\begin{equation}\label{Laplace}
\mathcal{L}_{\tilde{\lambda}}(p)=\mathbb{E}_{\lambda,\mathbf{w},G^{(N-1)}}\left\{e^{-p\mathbf{b}_{\lambda}^T\mathbf{b}_{\lambda}}
\delta\left(\lambda-\tilde{\lambda}\right)\right\}=
{\cal C}_{N}e^{-\frac{1}{2}\tilde{\lambda}^2}
\int dG^{(N-1)}\, e^{-\frac{1}{2}\Tr G^{(N-1)}G^{(N-1)T}}|\det{(\tilde{\lambda} I_{N-1}-G^{(N-1)})}|\,
\end{equation}
\[
\times \int d\mathbf{w} e^{-\frac{1}{2}\mathbf{w}^T\mathbf{w}-p\mathbf{w}^T\left[\left(\tilde{\lambda}\, I_{N-1}-G^{(N-1)}\right)^{T}\left(\tilde{\lambda}\, I_{N-1}-G^{(N-1)}\right)\right]^{-1}\,\mathbf{w}}
\]
As the integral over $\textbf{w}$ is Gaussian and $p>0$ it can be readily performed yielding the factor
\begin{equation}\label{detfact}
\det^{-1/2}\left(I_{N-1}+2p\left[\left(\tilde{\lambda}\, I_{N-1}-G^{(N-1)}\right)^{T}\left(\tilde{\lambda}\, I_{N-1}-G^{(N-1)}\right)\right]^{-1}\right)
\end{equation}
\[
=\frac{|\det{(\tilde{\lambda} I_{N-1}-G^{(N-1)})}|}{\det^{1/2}\left(2pI_{N-1}+\left(\tilde{\lambda}\, I_{N-1}-G^{(N-1)}\right)^{T}\left(\tilde{\lambda}\, I_{N-1}-G^{(N-1)}\right)\right)}\,
\]
where we have used $\det\left(\lambda\, I_{N-1}-G^{(N-1)}\right)^{T}=\det\left(\lambda\, I_{N-1}-G^{(N-1)}\right)$. Combining all the factors we finally see that  the characteristic function in question is proportional to the ensemble average of the ratio of determinants, cf. (\ref{DetGenbeta}) for $\beta=1$, which we also may present in an equivalent, but different form convenient for further evaluation:
\begin{equation}\label{Laplacereal}
\mathcal{L}_{\lambda}(p)= \frac{1}{\mathcal{N}} e^{-\frac{1}{2}\lambda^2} \mathcal{D}^{(2)}_{N-1,1}(\lambda,p),
\end{equation}
where as will be found below (see the Corollary 3.4) $\mathcal{N}=2^{N/2}\Gamma(N/2)$ and
\begin{equation}\label{Laplace1}
  \quad \mathcal{D}^{(2)}_{N-1,1}(\lambda,p)=\left\langle \frac{\det{\begin{pmatrix}
  0  & i(\lambda\, I_{N-1}-G) \\
  i(\lambda I_{N-1}-G^T) & 0
\end{pmatrix}}} {\det^{1/2}{\begin{pmatrix}
  \sqrt{2p}\,I_{N-1}  & i(\lambda I_{N-1}-G) \\
  i(\lambda I_{N-1}-G^T) & \sqrt{2p}\,I_{N-1}
\end{pmatrix}}} \right\rangle_{Gin_1, N-1}.
\end{equation}
Here the ensemble average is performed over the j.p.d. (\ref{Ginibre}) of real Ginibre matrices $G$ of the reduced size $(N-1)\times (N-1)$.
The problem of averaging the ratio of determinants in the above expression can be efficiently solved in the framework of the supersymmetry approach. The main steps of the corresponding procedure are presented in the following section. Interestingly, when implementing such an approach inverting the Laplace transform comes as a part of the procedure. In this way one recovers first (\ref{Main}) which by straightforward algebraic manipulations can be shown to be equivalent to (\ref{Mainpositive}).

\subsection{Supersymmetry approach to the ratio of determinants and proof of Theorem \ref{Theorem1}}\label{susy}
  In this section we evaluate
that ensemble average for real Ginibre matrices $G$ of size $N\times N$, with the main object of interest being
\begin{equation}\label{DNdef}
 \mathcal{D}_{N,1}^{(2)}(\lambda,p):=\left\langle \frac{\det{\begin{pmatrix}
  0  & i(\lambda\, I_{N}-G) \\
  i(\lambda I_{N}-G^T) & 0
\end{pmatrix}}} {\det^{1/2}{\begin{pmatrix}
  \sqrt{2p}\,I_{N}  & i(\lambda I_{N}-G) \\
  i(\lambda I_{N}-G^T) & \sqrt{2p}\,I_{N}
\end{pmatrix}}} \right\rangle_{Gin_1, N}\,.
\end{equation}
Our goal is to verify the following
\begin{proposition}
\begin{equation}\label{DNequal}
 \mathcal{D}_{N,1}^{(2)}(\lambda,p)=C_N\, e^{\lambda^2}\int_{\mathbb{R}_+}\frac{dt}{t^2}e^{-pt} e^{-\frac{\lambda^2}{2\left(1+\frac{1}{t}\right)}}\left(\frac{t}{1+t}\right)^{\frac{N+2}{2}}
\left[\Gamma(N+1,\lambda^2)-\lambda^2 \frac{t}{1+t}\Gamma(N,\lambda^2)\right]
 \end{equation}
 with the constant $C_N=\frac{1}{2^{N/2}\Gamma\left(\frac{N}{2}\right)}$.
\end{proposition}

{\bf Proof of the Theorem \ref{Theorem1}}. The {\bf Prop. 3.3} when combined with (\ref{Laplacereal}) immediately provides the proof of
 (\ref{Main}), hence of the Theorem \ref{Theorem1}. Namely, to arrive at (\ref{Main})
 one replaces $N\to N-1$ in  (\ref{DNequal}),  and substitutes it into (\ref{Laplacereal}).
Noting that the result assumes the form of a Laplace transform in variable $t$ makes its inversion trivial, and we recover the JPD $\mathcal{P}^{(r)}(t,\lambda)$ of the random variables $t=\mathbf{b}_{\lambda}^T\mathbf{b}_{\lambda}$ and the real eigenvalue $\lambda$ as is given in (\ref{Main}).\\[0.5ex]
{\bf Proof of the  Proposition 3.3}:
\begin{proof}
 Let $\Psi_1,\,\Psi_2,\,\Phi_1,\,\Phi_2$ be four column vectors with $N$ anticommuting components each. Using the standard rules of
Berezin integration one represents the numerator in the ratio (\ref{DNdef}) as a Gaussian integral
\begin{equation}\label{detGrassmanGauss}
\det{\begin{pmatrix}
  0  & i(\lambda\, I_N-G) \\
  i(\lambda I_N-G^T) & 0
\end{pmatrix}}
\end{equation}
\begin{equation*}
 = \int d\Psi_1\,d\Psi_2\,d\Phi_1\,d\Phi_2 \,\,  e^{-i(\Psi^T_1,\Phi^T_1)\begin{pmatrix}
  0  & \lambda\, I_N-G \\
  \lambda I_N-G^T & 0
\end{pmatrix}\begin{pmatrix}\Psi_2 \\ \Phi_2\end{pmatrix} }\,.
\end{equation*}
Now we further use a form of the standard Gaussian integral well-defined for any real-symmetric matrix $A$ and any positive $\epsilon>0$:
  \begin{equation}\label{intGauss}
  \int_{\mathbb{R}^N} d\mathbf{S} e^{-\frac{1}{2}\mathbf{S}^T\left(\epsilon I_N+iA\right)\mathbf{S}}=\frac{(2\pi)^{N/2}}{\det^{1/2}{\left(\epsilon I_N+iA\right)}}\,,
  \end{equation}
  where the integration goes over the vector $\mathbf{S}$ with $N$ real commuting components.
This allows to represent the denominator in (\ref{DNdef}) as  a Gaussian integral over two such vectors $\mathbf{S}_{1,2}$ :
\begin{equation}\label{invdetGauss}
\det^{-1/2}{\begin{pmatrix}
  \sqrt{2p}\,I_N  & i(\lambda I_N-G) \\
  i(\lambda I_N-G^T) & \sqrt{2p}\,I_N
\end{pmatrix}}
 \propto  \int_{\mathbb{R}^N} d\mathbf{S}_{1} \int_{\mathbb{R}^N}d\mathbf{S}_{2}\, e^{-\frac{1}{2}(\mathbf{S}_{1}^T,\mathbf{S}_{2}^T)\begin{pmatrix}
   \sqrt{2p}\,I_N   & i(\lambda\, I_N-G) \\
  i(\lambda I_N-G^T) &  \sqrt{2p}\,I_N
\end{pmatrix}\begin{pmatrix}\mathbf{S}_{1} \\ \mathbf{S}_{2}\end{pmatrix} }\,.
\end{equation}
where $\propto$ here and below stands for (temporaly) ignored  multiplicative constants (in general, $N-$dependent) whose product will be restored in the very end of the procedure.
After substituting the above representations to  (\ref{DNdef}) and rearranging in the exponent as $$\mathbf{S}_{1}^T G \mathbf{S}_{2}=\Tr \left(G\mathbf{S}_{2}\otimes \mathbf{S}_{1}^T\right) \quad \mbox{and} \quad \mathbf{\Psi}_{1}^T G \mathbf{\Phi}_{2}=-\Tr \left(G\mathbf{\Phi}_{2}\otimes \mathbf{\Psi}_{1}^T \right),$$ etc, where $M=\mathbf{a}\otimes \mathbf{b}^T$ stands for the matrix with entries $M_{ij}=a_ib_j$,
  one can easily perform the averaging over the real Ginibre
matrices by using the identity
\begin{equation}\label{Gauident}
\left\langle
e^{-\Tr(GA+G^TB)}
\right\rangle_{Gin_1}=e^{\frac{1}{2}\Tr(A^TA+B^TB+2AB)}\,.
\end{equation}
After the ensemble average is performed, there exists only one term in the exponential in the integrand which is quartic in anticommuting variables, and it
is of the form $\left(\Phi_1^T\Phi_2\right)\left(\Psi_1^T\Psi_2\right)$. The corresponding exponential factor is then represented as:
\begin{equation}\label{HSsimple}
e^{\left(\Phi_1^T\Phi_2\right)\left(\Psi_1^T\Psi_2\right)}=\frac{1}{2\pi}\int dq\, d\overline{q} e^{-|q|^2-q\left(\Psi_1^T\Psi_2\right)-\overline{q}\left(\Phi_1^T\Phi_2\right)}\,,
\end{equation}
where the formula above represents the simplest instance of what is generally known as the Hubbard-Stratonovich transformation.
After such a representation is employed, it allows to perform the (by now, Gaussian) integration over the anticommuting vectors explicitly, and reduce the whole expression to the integral over the two vectors $\mathbf{S}_{1,2}$ and over a single complex variable $q$:
\begin{equation}\label{intermed}
\mathcal{D}_{N,1}^{(2)}(\lambda,p)\propto \int dq\,d\overline{q}\, e^{-|q|^2}  \int d\mathbf{S}_{1} d\mathbf{S}_{2}\, e^{-\frac{1}{2}\sqrt{2p}(\mathbf{S}_{1}^T\mathbf{S}_{1}+\mathbf{S}_{2}^T\mathbf{S}_{2})-
\frac{i}{2}\lambda(\mathbf{S}_{1}^T\mathbf{S}_{2}+\mathbf{S}_{2}^T\mathbf{S}_{1})-\frac{1}{2}(\mathbf{S}_{1}^T\mathbf{S}_{1})(\mathbf{S}_{2}^T\mathbf{S}_{2})}
 \end{equation}
\[
\times \det{\begin{pmatrix}
  q\,I_N  & i\lambda I_N+\textbf{S}_1\otimes \textbf{S}_2^T \\
  i \lambda I_N + \textbf{S}_2\otimes \textbf{S}_1^T & \overline{q}\,I_N
\end{pmatrix}}
\]
A straightforward calculation shows that the determinant in the above expression is equal to
\[
(|q|^2+\lambda^2)^{N-2}
\left[(|q|^2+\lambda^2)^{2}+(|q|^2+\lambda^2)\left(-2i\lambda(\mathbf{S}_{1}^T\mathbf{S}_{2})-
(\mathbf{S}_{1}^T\mathbf{S}_{1})(\mathbf{S}_{2}^T\mathbf{S}_{2})\right)+\lambda^2\left((\mathbf{S}_{1}^T\mathbf{S}_{1})(\mathbf{S}_{2}^T\mathbf{S}_{2})-
(\mathbf{S}_{1}^T\mathbf{S}_{2})^2\right)\right]
\]
and we see that the integration over $q,\overline{q}$ is now easy to perform via using the polar coordinates:
\begin{equation}
\frac{1}{4\pi}\int dq\,d\overline{q}\, e^{-|q|^2}(|q|^2+\lambda^2)^N =e^{\lambda^2}\int_{\lambda^2}^{\infty}dt e^{-t}t^N:= e^{\lambda^2}\Gamma(N+1,\lambda^2)\,.
\end{equation}
 As to the remaining integrations,
 one may notice that the integrand depends only on the entries of a positive semidefinite real symmetric matrix
 \begin{equation}\label{QR}
 \hat{Q}=\begin{pmatrix}
  Q_1 & Q \\
  Q & Q_2
  \end{pmatrix}, \quad Q_1= \mathbf{S}_{1}^T\mathbf{S}_{1}, \, Q_2= \mathbf{S}_{2}^T\mathbf{S}_{2}, \, Q=\mathbf{S}_{1}^T\mathbf{S}_{2}
\end{equation}
A useful trick suggested in \cite{YF2002} in such a situation is to pass from the pair of vectors $(\mathbf{S}_{1}, \mathbf{S}_{2})$ to the matrix $ \hat{Q}$ as a new integration variable. Such change is non-singular for $N\ge 2$ and incurs a Jacobian factor proportional to $\det{\hat{Q}}^{(N-3)/2}$ (see the Appendix D of \cite{FyoStrahKaehler}). This finally brings $\mathcal{D}^{(2)}_{N,1}(\lambda,p)$  to the form
\begin{equation}\label{Laplace1a}
\mathcal{D}_{N,1}^{(2)}(\lambda,p)\propto e^{\lambda^2}\int_{\hat{Q}\ge 0} d\hat{Q} \, \det{\hat{Q}}^{(N-3)/2} e^{-\frac{1}{2}\sqrt{2p}(Q_1+Q_2)-i\lambda Q -\frac{1}{2}Q_1Q_2}
\end{equation}
\[\times  \left[\Gamma(N+1,\lambda^2)+\Gamma(N,\lambda^2)(-2i\lambda Q-Q_1Q_2)+\Gamma(N-1,\lambda^2)\lambda^2\det{Q}\right]\,.
\]
The next step requires employing a convenient parametrization of the integration domain defined via the inequalities  $Q_1\ge 0,\,  Q_2\ge 0, \, -\infty<Q<\infty$ and $Q_1Q_2\ge Q^2$ which ensure that $\hat{Q}=\begin{pmatrix}
  Q_1 & Q \\
  Q & Q_2
  \end{pmatrix}$ is a real symmetric positive semidefinite matrix.  First, it is easy to see that such domain can be parametrized by expressing the diagonal entries $Q_1$ and $Q_2$ in terms of two real coordinates $r\ge 0,\,-\infty<\theta<\infty$ chosen in such a way that $r=\left(\det \hat{Q}\right)^{1/2},\,\, \theta=\frac{1}{2}\ln{\left(Q_1/Q_2\right)}$.  By
   explicitly writing $Q_1=e^{\theta}\sqrt{r^2+Q^2},\, Q_2=e^{-\theta}\sqrt{r^2+Q^2}$ and evaluating the associated Jacobian we get in those coordinates
  $ d\hat{Q}:=dQ_1dQ_2dQ=2\, r\,dr \, dQ\,d\theta$. Although calculation in that parametrization is already quite convenient, it turns out that it becomes
  even shorter if one parametrizes the same domain in a related, but slightly less obvious way using instead the matrix entries $Q_1\ge 0$ and $-\infty<Q<\infty$ as new coordinates, complemented with $r=\left(\det \hat{Q}\right)^{1/2}\ge 0$, and expressing the remaining entry as $Q_2=\frac{r^2+Q^2}{Q_1}\ge 0$. This finally gives
   \begin{equation}\label{Qparam}
 \hat{Q}=\begin{pmatrix}
  Q_1 & Q \\
  Q & \frac{r^2+Q^2}{Q_1}
  \end{pmatrix}, \quad d\hat{Q}=2\, \frac{dQ_1}{Q_1}\, r\,dr \, dQ; \quad -\infty\le Q \le \infty,\, 0\le Q_1,r <\infty.
\end{equation}
and further changing $Q_1\to \sqrt{2p}Q_1$  brings (\ref{Laplace1a}) to the form
\begin{equation}\label{Laplace1bNEW}
\mathcal{D}_{N,1}^{(2)}(\lambda,p) \propto e^{\lambda^2}  \int_{\mathbb{R}_{+}} dr \, r^{N-2}e^{-\frac{r^2}{2}}\int_{\mathbb{R}} \frac{dQ}{\sqrt{2\pi}} e^{-i\lambda Q -\frac{1}{2}Q^2} \int_{\mathbb{R}_{+}}\frac{dQ_1}{Q_1}\,e^{-p Q_1-\frac{r^2+Q^2}{2Q_1}}
\end{equation}
\[\times  \left[\Gamma(N+1,\lambda^2)+\Gamma(N,\lambda^2)(-2i\lambda Q-Q^2-r^2)+\Gamma(N-1,\lambda^2)\lambda^2r^2\right]\]
where all integrals are well-defined and convergent for $p>0$; in particular, the latter one can be evaluated explicitly in terms of the Bessel function  of second kind as ( see  \cite{GR}, p.363)
\[
 \int_{\mathbb{R}_{+}}\frac{dQ_1}{Q_1}\,e^{-p Q_1-\frac{r^2+Q^2}{2Q_1}}=2K_0\left(\sqrt{2p\left(r^2+Q^2\right)}\right).
\]
In principle, one can demonstrate existence of a chain of integral identities which allows to perform the remaining integrations in (\ref{Laplace1bNEW}) explicitly without  changing the order of integrations. This way leads however to quite cumbersome intermediate formulas, and we proceed instead by changing the order in (\ref{Laplace1bNEW}) (which can be justified by Fubini's theorem) to
\begin{equation}\label{Laplace1b}
\mathcal{D}_{N,1}^{(2)}(\lambda,p) \propto e^{\lambda^2}  \int_{\mathbb{R}_{+}}\frac{dQ_1}{Q_1}e^{-p Q_1} \int_{\mathbb{R}_{+}} dr \, r^{N-2}e^{-\frac{r^2}{2}\left(1+\frac{1}{Q_1}\right)} \int_{\mathbb{R}} \frac{dQ}{\sqrt{2\pi}} e^{-i\lambda Q -\frac{1}{2}Q^2\left(1+\frac{1}{Q_1}\right)}
\end{equation}
\[\times  \left[\Gamma(N+1,\lambda^2)+\Gamma(N,\lambda^2)(-2i\lambda Q-Q^2-r^2)+\Gamma(N-1,\lambda^2)\lambda^2r^2\right]
\]
which allows to perform  the integrals over $Q$ and $r$ much more efficiently.
 Namely, introduce the function
\begin{equation}\label{defA}
A_{N}(s,t)=2\int_{\mathbb{R}_{+}} dr \, r^{N}e^{-\frac{r^2}{2}(s+1/t)}\int_{\mathbb{R}} \frac{dQ}{\sqrt{2\pi}} e^{-is\lambda Q -\frac{1}{2}Q^2\left(s+\frac{1}{t}\right)}
=2^{\frac{N+1}{2}}\Gamma\left(\frac{N+1}{2}\right)\left(s+\frac{1}{t}\right)^{-\frac{N+2}{2}}
e^{-\frac{s^2\lambda^2}{2\left(s+\frac{1}{t}\right)}}
\end{equation}
Then it is easy to see that after renaming $Q_1\to t$ the equation (\ref{Laplace1b}) can be rewritten as
\begin{equation}\label{Laplace1c}
\mathcal{D}_{N,1}^{(2)}(\lambda,p) \propto e^{\lambda^2}\lim_{s\to 1}\int_0^{\infty}\frac{dt}{t}e^{-pt}
\left\{\left[\Gamma(N+1,\lambda^2)+2\Gamma(N,\lambda^2)\frac{d}{ds}\right]A_{N-2}(s,t) +\Gamma(N-1,\lambda^2)\lambda^2 A_{N}(s,t)\right\}
\end{equation}
\begin{equation}
\propto \,\, e^{\lambda^2}\int_{\mathbb{R}_+}\frac{dt}{t}e^{-pt} e^{-\frac{\lambda^2}{2\left(1+\frac{1}{t}\right)}}\left(1+\frac{1}{t}\right)^{-\frac{N+2}{2}}
\end{equation}
\[
\times \left[\Gamma(N+1,\lambda^2)\left(1+\frac{1}{t}\right)
-\Gamma(N,\lambda^2)\left(N+\lambda^2\left(1+\frac{1}{1+t}\right)\right)+\lambda^2 \Gamma(N-1,\lambda^2)(N-1)\right]
\]
Now by using the relation $\Gamma(N+1,\lambda^2)=e^{-\lambda^2}\lambda^{2N}+N\Gamma(N,\lambda^2)$ one can see that
\[
\Gamma(N+1,\lambda^2)-(N+\lambda^2)\Gamma(N,\lambda^2)+\lambda^2(N-1)\Gamma(N-1,\lambda^2)=0
\]
implying finally
\begin{equation}\label{DNFin}
\mathcal{D}_{N,1}^{(2)}(\lambda,p) = C_N e^{\lambda^2}\int_{\mathbb{R}_+}\frac{dt}{t^2}e^{-pt} e^{-\frac{\lambda^2}{2\left(1+\frac{1}{t}\right)}}\left(\frac{t}{1+t}\right)^{\frac{N+2}{2}}\left[e^{-\lambda^2}\lambda^{2N}+\Gamma(N,\lambda^2)\left(N-\lambda^2 \frac{t}{1+t}\right)\right]
\end{equation}
for some real constant $C_N$. Simple manipulations with incomplete $\Gamma-$function (\ref{Gamma}) show that this is equivalent to (\ref{DNequal}).  To establish the value for the constant one can use, for example, the limit $p\to \infty$ where according to the definition (\ref{DNdef})
\begin{equation}\label{limit1}
\lim_{p\to \infty} (2p)^{N/2}\mathcal{D}_{N,1}^{(2)}(\lambda,p) =\left\langle \det^2{(\lambda I_N-G)} \right\rangle_{Gin_1, N}
\end{equation}
On the other hand, performing $p\to \infty$ limit in (\ref{DNFin}) gives after a simple calculation
\begin{equation}\label{limit2}
\lim_{p\to \infty} (2p)^{N/2}\mathcal{D}_{N,1}^{(2)}(\lambda,p) = C_N \, 2^{N/2}\Gamma\left(\frac{N}{2}\right)\left[\lambda^{2N}+e^{\lambda^2}\,N\,\Gamma(N,\lambda^2) \right]
\end{equation}
The definition of the left-hand side implies that the coefficient in front of $\lambda^{2N}$ must be equal to unity, giving finally
\begin{equation}\label{const}
 C_N = \frac{1}{2^{N/2}\Gamma\left(\frac{N}{2}\right)}, \quad \left\langle \det^2{(\lambda I_N-G)} \right\rangle_{Gin_1, N}=\lambda^{2N}+e^{\lambda^2}\,N\,\Gamma(N,\lambda^2)
\end{equation}
\end{proof}

\begin{corollary}
The normalization constant $\mathcal{N}$ in Eq.(\ref{Laplacereal}) is given by $\mathcal{N}=2^{N/2}\Gamma(N/2)$.
\end{corollary}
\begin{proof}

To establish the value of the constant $\mathcal{N}$ we consider the limit $p\to 0$ in both sides of (\ref{Laplacereal}). By the very definition of the Laplace transform $\mathcal{L}_{\lambda}(p)$ its value at $p=0$ must be equal to the mean density of real eigenvalues for the real Ginibre ensemble
given in (\ref{Edelman}). On the other hand, for $p=0$ the integration over $t$ in (\ref{DNequal}) can be easily performed introducing $u=|\lambda|\sqrt{\frac{t}{1+t}}$
as new integration variable. One gets in this way:
\begin{equation}\label{EKS}
\mathcal{D}_{N,1}^{(2)}(\lambda,0)=\frac{2}{2^{N/2}\Gamma\left(\frac{N}{2}\right)}
\left[e^{\frac{\lambda^2}{2}}\Gamma\left(N,\lambda^2\right)+|\lambda|^{N}\int_0^{|\lambda|} e^{-\frac{u^2}{2}} u^{N-1}\,du\right]
\end{equation}
To get $\mathcal{D}_{N-1,1}^{(2)}(\lambda,0)$ featuring in the right-hand side of (\ref{Laplacereal}) replace in the above $N\to N-1$ and use $\Gamma(N)=\frac{2^{N-1}}{\sqrt{\pi}} \Gamma\left(\frac{N}{2}\right)\Gamma\left(\frac{N+1}{2}\right)$. Multiplying with $e^{-\lambda^2/2}$ and comparing
with the left-hand side gives the value for the constant $\mathcal{N}$.
\end{proof}

\section{Proof of Theorem \ref{Theorem2}}
Our approach to complex Ginibre matrices follows essentially the same steps as for the real case, with very little modifications, and we
 only briefly indicate necessary changes.   Similarly to (\ref{reductreal}), suppose that a complex-valued matrix $G^{(N)}$ has only non-degenerate eigenvalues, and assuming it has an eigenvalue $z$ (in general, complex) it can be represented as (see e.g. Sec. 6 of \cite{FyoKhor2007})
\begin{equation}\label{reduct}
G^{(N)}=R
\begin{pmatrix}
  z  & {\bf w}^{*} \\
  0 & G^{(N-1)}
\end{pmatrix} R^{*}
\end{equation}
for some complex $(N-1)$- component column vector $\mathbf{w}$ and a complex matrix $ G^{(N-1)}$ of size  $(N-1)\times (N-1)$, whereas the matrix $R$ is Hermitian and unitary: $R^*=R$ and $RR^{*}=R^2=I_N$. Considering the right eigenvectors of the 'rotated' matrix $\tilde{G}^{(N)}$ gives  $\mathbf{\tilde{x}}_R=\mathbf{e}_1$, whereas $\mathbf{\tilde{x}}^*_L=(1,\mathbf{b}^*_{z})$, where $\mathbf{b}^*_{z}$ is an $(N-1)-$component complex row vector. This implies ${\cal \tilde{O}}_{z}=1+\mathbf{b}_{z}^*\mathbf{b}_{z}$, whereas the vector $\mathbf{b}_{z}^*$ is readily expressed in terms of the resolvent of $\tilde{G}^{(N)}$ as  $\mathbf{b}_{z}^*= \mathbf{w}^*\left(z\, I_{N-1}-G^{(N-1)}\right)^{-1}$
implying
\begin{equation}\label{bvecnorm}
 \mathbf{b}_{z}^*\mathbf{b}_{z}=\mathbf{w}^*\left[\left(\overline{z}\, I_{N-1}-G^{(N-1)*}\right)\left(z\, I_{N-1}-G^{(N-1)}\right)\right]^{-1}\,\mathbf{w}
\end{equation}
Now we exploit the analogue of Prop. (\ref{t:jpdfreal})
\begin{proposition}[see Appendix B of \cite{FyoKhor2007}]\label{t:jpdfcomp}

Assume that the j.p.d. of $G$ is given by (\ref{Ginibre}) with $\beta=2$, and apply the transformation (\ref{reduct}). Then the joint probability density
of elements of the matrix  $\tilde{G}$ is given by

\begin{equation}\label{volchangeComplex}
{\cal P}(\tilde{G})d\tilde{G}={\cal C}_{2,N}|\det{(z I_{N-1}-G^{(N-1)})}|^2\,e^{-\left(|z|^2+\mathbf{w}^*\mathbf{w}+\Tr G^{(N-1)}G^{(N-1)*}\right)}\,dzd\overline{z} \,d\mathbf{w}d\mathbf{w}^*\,dG^{(N-1)} \, .
\end{equation}
with some normalization constant ${\cal C}_{2,N}$.
\end{proposition}
Using the above j.p.d. to calculate the Laplace transform of the probability density for the random variable
$\mathbf{b}_{z}^*\mathbf{b}_{z}$  for a fixed value of $z$ we arrive after standard manipulations at representing it as
 the expectation of the ratios of the determinants of the form
\begin{equation}\label{LaplaceComplex}
\mathcal{L}_{z}(p)=  \frac{1}{\pi \Gamma(N)}e^{-|z|^2} {\cal D}^{(2)}_{N-1,2}(z,p)
\end{equation}
where we introduced the notation,  cf. (\ref{DetGenbeta}) for $\beta=2$,
\begin{equation}\label{DetComplex}
 {\cal D}^{(L)}_{N,2}(z,p)=
\left\langle \frac{\det^L{\begin{pmatrix}
  0  & i(z\, I_{N}-G) \\
  i(\overline{z} I_{N}-G^*) & 0
\end{pmatrix}}} {\det{\begin{pmatrix}
  \sqrt{p}\,I_{N}  & i(z I_{N}-G) \\
  i(\overline{z} I_{N}-G^*) & \sqrt{p}\,I_{N}
\end{pmatrix}}} \right\rangle_{Gin_2, N}
\end{equation}
with averaging performed over the j.p.d. (\ref{Ginibre}) of complex Ginibre matrices $G$ of the  size $N\times N$.
Note, that the value of the constant normalization factor in (\ref{LaplaceComplex}) is found {\it aposteriori} exactly in the same way
as in the real case, by comparing  the known expression for the
mean density of complex eigenvalues (\ref{denscomp}) (coinciding with $\mathcal{L}_{z}(0)$) and the corresponding limit in the right-hand side of
(\ref{LaplaceComplex}).

 In the general case evaluating $\mathcal{D}_{N,2}^{(L)}(z,p)$ for integer $L$ can be done essentially by the same  supersymmetry method  which was used in section (\ref{susy}), with obvious necessary modifications imposed by symmetries.
In particular, presence of higher powers of the determinants in the numerator of  (\ref{DetComplex}) necessitates to use $L$ sets of anticommuting vectors for their representation, making the resulting integral representation in our version of the supersymmetry method significantly more cumbersome than in the real case.  In the most relevant case $L=2$ and a special choice of the spectral parameter $z=0$ the expected value of the ratio featuring in the right-hand side of  (\ref{DetComplex}) can be relatively easily extracted as a special limiting case of a more general object evaluated in \cite{FyoAkk2003} or, in a different way, in
 \cite{FyoStrahChiral}.  The corresponding calculation is sketched in the first part of the next section. The supersymmetry approach for $z\ne 0$ works along exactly the same general lines as in the real case, but is somewhat more involved technically. The corresponding calculation is outlined in the second part of the next section.

\subsection{Evaluation of  (\ref{DetComplex}) for $L=2$ }
\subsubsection{ Evaluation of  (\ref{DetComplex}) for $L=2$ and $|z|=0$}
In the special case $|z|=0$ an integral representation of the averaged ratio of determinants featuring in (\ref{LaplaceComplex}) which we find most
convenient for our purposes was derived in \cite{FyoStrahChiral}, see Eq.(29) there. Actually, our object arises as a particular case $n_f=2, n_b=1$ of that formula, identifying $x_b=2\sqrt{p}$ and considering a special limit $X_f=0$ (the latter limit is highly degenerate, and it is easier to perform it directly in eq.(25), and then rederive (29)). In this way we arrive at representing $\mathcal{D}_{N,2}^{(L=2)}(z=0,p)$ as
\begin{equation}\label{comp1}
\left\langle \frac{\det^2{\begin{pmatrix}
  0  & iG \\
  i G^* & 0
\end{pmatrix}}} {\det{\begin{pmatrix}
  \sqrt{p}\,I_{N}  & i G) \\
  iG^* & \sqrt{p}\,I_{N}
\end{pmatrix}}} \right\rangle_{Gin_2}
\propto \int_{0}^{\infty} dR_1 \int_{0}^{\infty} dR_2 (R_1 R_2)^{N-1} e^{-(R_1+R_2)}(R_1-R_2)^2
\end{equation}
\[
\times \int_{0}^{\infty} dR R^{N-1} e^{-R} K_0\left(2\sqrt{pR}\right) (R-R_1)(R-R_2)
\]
Now the integrals over $R_1$ and $R_2$ can be readily evaluated, with the result being simply
\[
\left\langle \frac{\det^2{\begin{pmatrix}
  0  & iG \\
  i G^* & 0
\end{pmatrix}}} {\det{\begin{pmatrix}
  \sqrt{p}\,I_N  & i G) \\
  iG^* & \sqrt{p}\,I_N
\end{pmatrix}}} \right\rangle_{Gin_2}
\propto \left[\mathcal{J}_{N+1}(p)-2(N+1)\mathcal{J}_{N}(p)+N(N+1)\mathcal{J}_{N-1}(p)\right]
\]
where we introduced the notation
\[
\mathcal{J}_N(p)=\int_{0}^{\infty} dR R^{N} e^{-R} K_0\left(2\sqrt{pR}\right)\,.
 \]
Further employing a well-known integral representation for the Bessel function of the second kind
\[
K_0\left(2\sqrt{pR}\right)=\frac{1}{2}\int_0^{\infty}\frac{dt}{t}e^{-pt-\frac{R}{2t}}
\]
shows that
\begin{equation}
\mathcal{J}_N(p)=\frac{1}{2}\int_0^{\infty}\frac{dt}{t}e^{-pt}\int_{0}^{\infty} dR R^{N} e^{-R\left(1+\frac{1}{t}\right)}=\frac{N!}{2}\int\frac{dt}{t}e^{-pt}\frac{1}{\left(1+\frac{1}{t}\right)^{N+1}}
\end{equation}
We finally conclude that
\begin{equation}\label{LaplaceComplexFin}
\mathcal{D}_{N,2}^{(L=2)}(z=0,p)
 = N\,(N+1)! \int_0^{\infty}\frac{dt}{t}e^{-pt}\frac{1}{\left(1+\frac{1}{t}\right)^N}\left(1-\frac{1}{1+\frac{1}{t}}\right)^2
\end{equation}
which after substituting to the Laplace transform (\ref{LaplaceComplex}) is equivalent to (\ref{MainComp}).

\subsubsection{ Evaluation of  (\ref{DetComplex}) for $L=2$ and $|z|\ne 0$ by supersymmetry approach}

Our goal is to verify the following
\begin{proposition}
\begin{equation}\label{LaplacecompFINAL}
\mathcal{D}^{(2)}_{N,2}(z,p) = \frac{1}{(N-1)!}  e^{2|z|^2} \int_{\mathbb{R}_{+}}\frac{dt}{t}\, e^{-p t}\,\,\frac{e^{-|z|^2\frac{t}{1+t}}}{(1+t)^2}
\left(\frac{t}{1+t}\right)^{N}\left(D^{(N+1)}_1+|z|^2\frac{D_2^{(N+1)}}{1+t}+|z|^4\frac{d_1^{(N+1)}}{(1+t)^2}\right)
\end{equation}
where the entering quantities were defined in equations (\ref{LaplacecompFINALcoeffA})-(\ref{LaplacecompFINALcoeffD}).
\end{proposition}

The proof is very similar to the real case, and is outlined below.
\begin{proof}
One starts with using two copies of the set of four anticommuting vectors, namely $\Psi_{A1},\,\Psi_{A2},\,\Phi_{A1},\,\Phi_{A2}$ and $\Psi_{B1}\,\Psi_{B2},\,\Phi_{B1},\,\Phi_{B2}$, to represent separately two determinants in the numerator via Gaussian integrals, see
(\ref{detGrassmanGauss}). At the same time,  one needs two commuting vectors ${\bf S}_1, {\bf S}_2$ with $N$ {\it complex-valued} components to represent the denominator:
\begin{equation}\label{invdetGaussCOMP}
\det^{-1}{\begin{pmatrix}
  \sqrt{p}\,I_N  & i(z I_N-G) \\
  i(\overline{z} I_N-G^*) & \sqrt{p}\,I_N
\end{pmatrix}}
 \propto  \int_{\mathbb{C}^N} d\mathbf{S}_{1} d\mathbf{S}^*_{1}\int_{\mathbb{C}^N}d\mathbf{S}_{2}\, d\mathbf{S}_{2}^* e^{-(\mathbf{S}_{1}^*,\mathbf{S}_{2}^*)\begin{pmatrix}
   \sqrt{p}\,I_N   & i(z\, I_N-G) \\
  i(\overline{z} I_N-G^*) &  \sqrt{p}\,I_N
\end{pmatrix}\begin{pmatrix}\mathbf{S}_{1} \\ \mathbf{S}_{2}\end{pmatrix} }\,.
\end{equation}

 The ensemble averaging is performed by exploiting $\beta=2$ analogue
of (\ref{Gauident})
\begin{equation}\label{GauidentCOMP}
\left\langle
e^{-\Tr(GA+G^*B)}
\right\rangle_{Gin_2}=e^{\Tr(AB)}\,.
\end{equation}
Performing the average one collects all terms in the exponential which are quartic in anticommuting variables, e.g. $\left(\Psi_{A1}^T\Phi_{A1}\right)\left(\Psi_{A2}^T\Phi_{A2}\right)$, etc.. The corresponding exponential factor can be then represented
via a matrix version of the Hubbard-Stratonovich transformation  generalizing (\ref{HSsimple}):
\begin{equation}\label{HSCOMP}
\exp{{\Tr}\left(\begin{array}{cc}\Psi_{A1}^T\Phi_{A1} &\Psi_{A1}^T\Phi_{B1}\\ \Psi_{B1}^T\Phi_{A1} & \Psi_{B1}^T\Phi_{B1}
\end{array}\right)\left(\begin{array}{cc}\Psi_{A2}^T\Phi_{A2} &\Psi_{A2}^T\Phi_{B2}\\ \Psi_{B2}^T\Phi_{A2} & \Psi_{B2}^T\Phi_{B2}
\end{array}\right)}
\end{equation}
\[
=\int d\hat{Q}_F\, d\hat{Q}_F^*\, \exp{-\Tr \left(\hat{Q}_F^*\, \hat{Q}_F\right)-\Tr\left(\hat{Q}_F\left(\begin{array}{cc}\Psi_{A1}^T\Phi_{A1} &\Psi_{A1}^T\Phi_{B1}\\ \Psi_{B1}^T\Phi_{A1} & \Psi_{B1}^T\Phi_{B1}
\end{array}\right)+\hat{Q}_F^*\left(\begin{array}{cc}\Psi_{A2}^T\Phi_{A2} &\Psi_{A2}^T\Phi_{B2}\\ \Psi_{B2}^T\Phi_{A2} & \Psi_{B2}^T\Phi_{B2}
\end{array}\right)\right)}\,,
\]
 with $\hat{Q}_F,\hat{Q}_F^*$ being a pair of general $2\times 2$ complex conjugate matrices. This trick allows to integrate out the vectors with
 anticommuting component completely. The analogue of (\ref{intermed}) takes the form
\begin{equation}\label{intermedCOMP}
\mathcal{D}_{N,2}^{(2)}(z,p)\propto \int d\hat{Q}_F\,d\hat{Q}_F^*\, e^{-\Tr\hat{Q}_F\,\hat{Q}_F^*} \int d\mathbf{S}_{1} d\mathbf{S}^*_{1} d\mathbf{S}_{2} d\mathbf{S}_{2}^*\, e^{-\sqrt{p}(\mathbf{S}_{1}^*\mathbf{S}_{1}+\mathbf{S}_{2}^*\mathbf{S}_{2})-
iz\mathbf{S}_{1}^*\mathbf{S}_{2}-i\overline{z}\mathbf{S}_{2}^*\mathbf{S}_{1}-(\mathbf{S}_{1}^*\mathbf{S}_{1})(\mathbf{S}_{2}^*\mathbf{S}_{2})}
 \end{equation}
\[
\times \det{\begin{pmatrix}
  \hat{Q}_F\otimes I_N  &  \left(i\,z I_N+\textbf{S}_1\otimes \textbf{S}_2^*\right)\otimes I_2 \\
 \left( i\, \overline{z} I_N + \textbf{S}_2\otimes \textbf{S}_1^*\right)\otimes I_2  & \hat{Q}_F^*\otimes I_N
\end{pmatrix}}
\]
At the next step one can simplify the above expression by employing the singular value decomposition $\hat{Q}_F=U\mbox{diag}(\sqrt{R_{F1}},\sqrt{R_{F2}})V^{*}, \hat{Q}_F=V\mbox{diag}(\sqrt{R_{F1}},\sqrt{R_{F2}})U^{*}$ with unitary $U,V$ and
$R_{F1},R_{F2}\ge 0$ and replacing the integration over complex vectors $\textbf{S}_{1,2}$ with one over the Hermitian positive semidefinite matrix (cf. (\ref{QR})):
 \begin{equation} \label{QC}
 \hat{Q}=\begin{pmatrix}
  Q_1 & Q \\
  \overline{Q} & Q_2
  \end{pmatrix}, \quad Q_1= \mathbf{S}_{1}^*\mathbf{S}_{1}, \, Q_2= \mathbf{S}_{2}^*\mathbf{S}_{2}, \, Q=\mathbf{S}_{1}^*\mathbf{S}_{2}, \,  \overline{Q}=\mathbf{S}_{2}^*\mathbf{S}_{1}
\end{equation}
After straightforward algebraic manipulations this allows to represent  (\ref{intermedCOMP}) for $N\ge 2$ as
\begin{equation}\label{intermedCOMP1}
\mathcal{D}_{N,2}^{(2)}(z,p)\propto \int  d R_{F1}  d R_{F2} (R_{F1}-R_{F2})^2\, e^{-(R_{F1}+R_{F2})} \int_{\hat{Q}\ge 0} d{\hat{Q}}\, \det^{N-2}{\hat{Q}}
e^{-\sqrt{p}(Q_{1}+Q_{2})-i\left(z\, Q+\overline{z}\overline{Q}\right)-Q_1Q_2}
 \end{equation}
\[
\times \prod_{k=1}^2 \left(R_{Fk}+|z|^2\right)^{N-2}\left[ \left(R_{Fk}+|z|^2\right)^{2}- \left(R_{Fk}+|z|^2\right)\left(Q_1Q_2+i\left(z\, Q+\overline{z}\overline{Q}\right)\right)+|z|^2\left(Q_1Q_2-|Q|^2\right)\right]
\]
The Hermitian matrix $\hat{Q}\ge 0$ can be parametrized very similarly to (\ref{Qparam}). Namely, writing for the complex variable $Q=\rho\,e^{i\phi}$ and using $r=\left(\det{\hat{Q}}\right)^{1/2}\ge 0$
together with $Q_1\ge 0$ as the coordinates, the domain of integration is parametrized by matrices
\begin{equation}\label{QparamCOMP}
 \hat{Q}=\begin{pmatrix}
  Q_1 & \rho\,e^{i\phi} \\
   \rho\,e^{-i\phi} & \frac{r^2+\rho^2}{Q_1}
  \end{pmatrix}, \quad d\hat{Q}=2\, \frac{dQ_1}{Q_1}\, r\,dr \, \rho\,d\rho d\phi ; \quad 0 \le \phi \le 2\pi,\, 0\le Q_1,\rho,r <\infty.
\end{equation}
In this way we arrive at an analogue of (\ref{Laplace1b}):
\begin{equation}\label{Laplace1bcomp}
\mathcal{D}^{(2)}_{N,2}(z,p) =  e^{2|z|^2}  \int_{\mathbb{R}_{+}}\frac{dQ_1}{Q_1}e^{-p Q_1} \int_{\mathbb{R}_{+}} dr \,
 r^{2N-
1}e^{-r^2\left(1+\frac{1}{Q_1}\right)} \int_{\mathbb{R}} \frac{d\rho  \rho}{\sqrt{2\pi}} e^{-\rho^2\left(1+\frac{1}{Q_1}\right)}
\int _0^{2\pi} \frac{d\phi}{2\pi} e^{-2i |z| \rho \cos{\phi}}
\end{equation}
\[\times  \left\langle \left[R^2_{F1}-R_{F1}(2i|z| \rho \cos{\phi}+\rho^2+r^2)+|z|^2r^2\right]\left[R^2_{F2}-R_{F2}(2i|z| \rho \cos{\phi}+\rho^2+r^2)+|z|^2r^2\right]\right\rangle_{FF}
\]
where we denoted
\begin{equation}\label{FFint}
\left\langle (\ldots)\right\rangle_{FF}=\int_{|z|^2}^{\infty} dR_{F1}\int_{|z|^2}^{\infty} dR_{F2}(R_{F1}-R_{F2})^2e^{-R_{F1}-R_{F2}} (\ldots)
\end{equation}
The integrals over $\rho$, $r$, $dR_{F1}$ and $dR_{F2}$ can be performed. Namely, denoting $r^2=R_B$ and $\rho^2=R$ we introduce
a function, cf. (\ref{defA}),
\begin{equation}\label{defB}
B_{N}(s,t):=\int_{\mathbb{R}_{+}} dR_B \, R_B^{N}e^{-R_B(s+1/t)}\int_{\mathbb{R}} dR  e^{-R\left(s+\frac{1}{t}\right)}\int _0^{2\pi} \frac{d\phi}{2\pi} e^{-2i |z| s \sqrt{R} \cos{\phi}}
\end{equation}
\[
= \Gamma\left(\frac{N+1}{2}\right)\left(s+\frac{1}{t}\right)^{-(N+2)}
e^{-\frac{s^2|z|^2}{s+\frac{1}{t}}}
\]
After further denoting $e^{|z|^2}\Gamma(N,\lambda^2):= \gamma\left(N,|z|^2\right)$ we then define the functions
\begin{equation}\label{defc}
C_{N}^{(k,l)}(|z|^2)=\gamma\left(N-k+3,|z|^2\right)\gamma\left(N-l+1,|z|^2\right)+\gamma\left(N-l+3,|z|^2\right)\gamma\left(N-k+1,|z|^2\right)
\end{equation}
\[
-2\gamma\left(N-k+2,|z|^2\right)\gamma\left(N-l+2,|z|^2\right)
\]
and renaming $Q_1\to t$ we notice that (\ref{Laplace1bcomp}) can be represented, after restoring the normalization constants, as
\begin{equation}\label{Laplace2bcomp}
\mathcal{D}^{(2)}_{N,2}(z,p) = \frac{1}{2(N-1)!(N-2)!}\lim_{s\to 1}  \int_{\mathbb{R}_{+}}\frac{dt}{t}e^{-p t} \left\{ C_{N}^{(0,0)}(|z|^2)B_{N-2}(s,t)+2C_{N}^{(0,1)}(|z|^2)\frac{dB_{N-2}}{ds}\right.
\end{equation}
\[
\left.+2|z|^2C_{N}^{(0,2)}(|z|^2)B_{N-1}(s,t)+C_{N}^{(1,1)}(|z|^2)\frac{d^2B_{N-2}}{ds^2}+2|z|^2C_{N}^{(1,2)}(|z|^2)\frac{dB_{N-1}}{ds}+
|z|^4C_{N}^{(2,2)}(|z|^2)B_{N}(s,t)\right\}
\]
Substituting (\ref{defB}) to the above and taking the limit, the expression can be further simplified with the help of symbolic manipulations using Wolfram Mathematica and is fianlly represented as (\ref{LaplacecompFINAL}).
In particular, it is easy to see that $D^{(N+1)}_1|_{|z|=0}=N(N+1)d_1^{(N+1)}|_{|z|=0}=N!(N+1)!$ so that (\ref{LaplacecompFINAL}) at $|z|=0$ indeed reproduces (\ref{LaplaceComplexFin}).
\end{proof}

\end{document}